\documentclass[10pt,journal,twocolumn]{IEEEtran}

\usepackage{amsmath}
\usepackage{amsfonts}
\usepackage{amssymb}

\usepackage{upref}
\usepackage{theorem}
\usepackage{graphicx}
\usepackage{psfrag}
\usepackage{color}

\newtheorem{theorem}{Theorem}
\newtheorem{lemma}{Lemma}

\newtheorem{definition}{Definition}

\newcommand{\delq}[1]{}
\newcommand{\wiretap}{{\cal WT}}

\def\bX{{\bf X}}

\def\bY{{\bf Y}}

\def\bZ{{\bf Z}}
\def\bV{{\bf V}}
\def\bU{{\bf U}}

\def\bH{{\bf H}}

\def\bW{{\bf W}}
\def\bV{{\bf V}}
\def\bE{{\bf E}}
\def\bF{{\bf F}}
\def\bK{{\bf K}}
\def\bGa{{\mathcal{N}}}
  
\def\Capac{{\cal C}_c}
\def\bgamma{{\boldsymbol{\gamma}}}
\def\blambda{{\boldsymbol{\lambda}}}
\def\bLambda{{\boldsymbol{\Lambda}}}

\title{Oblivious Transfer over Wireless Channels} 
\author{
Jithin~Ravi\authorrefmark{1},
        Bikash~Kumar~Dey\authorrefmark{1},
        Emanuele~Viterbo\authorrefmark{2}
\thanks{This paper was presented in part at the Information
Theory Workshop, ITW 2015, Jerusalem. The work of Jithin~R and B.~K.~Dey was supported by Department of Science and Techonolgy under grant SB/S3/EECE/057/2013
and by Information Technology Research Academy under grant ITRA/15(64)/Mobile/USEAADWN/01. The work of E. Viterbo was supported by the Australian Research Council
through the Discovery Project under grant DP130100336. }

\thanks{\authorrefmark{1}Jithin~R. and B.~K.~Dey are with the Department of Electrical Engineering at
IIT Bombay, Mumbai, INDIA-400076. Email:\{rjithin,bikash\}@ee.iitb.ac.in}

\thanks{\authorrefmark{2}E.~Viterbo is with the Department of Electrical \& Computer Systems Engineering at Monash University, Australia. Email: emanuele.viterbo@monash.edu.}
}

\begin{document}
\maketitle
\begin{abstract}
We consider the problem of oblivious transfer (OT) over OFDM and MIMO 
wireless communication systems where only the receiver knows
the channel state information. The sender and receiver also have
unlimited access to a noise-free real channel.
Using a physical layer approach, based 
on the properties of the noisy fading channel, we propose a scheme 
that enables the transmitter to send obliviously one-of-two files, 
i.e., without knowing which one has been actually requested by the receiver,
while also ensuring that the receiver does not get any information about
the other file.
\end{abstract}
\section{Introduction}
Consider a movie server, or a server of medical database. A subscriber
wants a specific item (a movie, or information about a specific disease)
without the server being able to know which item is desired by the subscriber.
The subscriber is also not allowed to gain any significant information
about any other item. This is an example of oblivious transfer. 

In one-out-of-two string oblivious transfer (OT), 
one party, Alice, has two files and the other party, Bob, wants one of 
these files.
Bob needs to obtain the required file without Alice finding out the identity of the file chosen by him. Bob should also not be able to recover any significant
information about the other file. 
Alice and Bob are assumed to be ``honest but curious'' participants -
they follow the agreed protocol
but are also curious to gain additional knowledge of the other's
data from their own observations during the protocol~\cite{Ahlswede,Nascimento}. 

\delq{
BEC 1/2 OT capacity: Imai, Morozov, Nacimento 2006

Nacimento, Winter 2006/2008 : Capacity >0 for a class of channels/sources

BEC(\epsilon) capacity : Ahlswede & Csiszar, also upper bound on capacity

BSC protocol : Crepeau 1997

AWGN: Isaka 2010, Isaka 2009 (latice) 

Kilian 1988: OT is a primitive for 2-party circuit evaluation

Rabin 1981: single file OT  

Imai, Nacimento, Morozov  2006: BEC 1/2

 S. Wiesner, ”Conjugate coding,” 1970 - first time proposed OT, never published
}

OT has been studied in various forms for some time in 
cryptography~\cite{Rabin1981,Crepeau1997}. It is a special
case of secure function computation problems, where multiple parties
want to compute a function without revealing additional information
about their data to other parties.
It was shown by Kilian~\cite{Kilian1988} that an OT protocol can be
used as a subroutine to devise a protocol for two-party secure function 
computation for any function that is representable by a boolean circuit.

It is well known that OT can not be performed only by interactive 
communication over a noise-free channel.
The OT is thus studied with a noisy channel as a critical resource in
addition to unlimited access to a noise-free channel.
The OT capacity is the largest length of file that can be transferred, per use
of the noisy channel, between Alice and Bob.
In \cite{Ahlswede}, \cite{Nascimento}, 
one-out-of-two string OT has been studied when
the noisy channel between Alice and Bob is a Discrete Memoryless Channel (DMC).
An upper bound for the OT capacity of a DMC was given in \cite{Ahlswede}
and it was shown that the given upper bound is achievable by a simple
scheme for binary erasure channels (BEC). Multi-user variants of OT have 
been studied over broadcast erasure channels in 
\cite{MishraDPDisit14,MishraDPDitw-invited14}.

One-out-of-two string OT has been considered in the context of AWGN channels
in \cite{Isaka}, where a protocol was proposed. The case of fast fading wireless
channels has also been discussed in \cite{Isaka}, where the fading
state varies in each transmission and is not known to the transmitter
or the receiver. Under such assumption, the channel can be modeled
by the conditional probability distribution $p_{Y|X}$ with the channel state
marginalized.
The fading state does not directly provide any additional advantage in OT here,
other than through its influence on $p_{Y|X}$.
The OT capacity is not known for many important channels including AWGN
and binary symmetric channels.

In this paper, we consider OT over two
classes of wireless slow-fading channels: orthogonal frequency division multiplexing
(OFDM) channel and multiple input multiple output (MIMO) channel,
where the fading state information is available only
at the receiver (CSIR), \cite{Tse_book}. 
Channels with CSIR (Fig.~\ref{fig:otbasic}) have not been
considered for OT before to the best of our knowledge.
CSIR is a common assumption in wireless communication which can be made
when the coherence block length $n$ is sufficiently large. 
We allow an interactive protocol to run over $n$ uses of the channel
during which the channel state remains fixed,
and in that period the noise-free channel can be used any finite
number of times.
In other words, we assume that one run of the OT protocol is completed 
in one coherence block. 
However, following common principle of rate-adaptation used in many wireless
communication models, the OT rate may vary from block to block depending
on the channel state.
As we will see in our schemes, the knowledge
of the state only at the receiver is the key to some interesting techniques for OT. Our techniques have the flavor of the protocol for BECs~\cite{Ahlswede}.

\begin{figure}[htbp]
\centering
\includegraphics[width=\columnwidth]{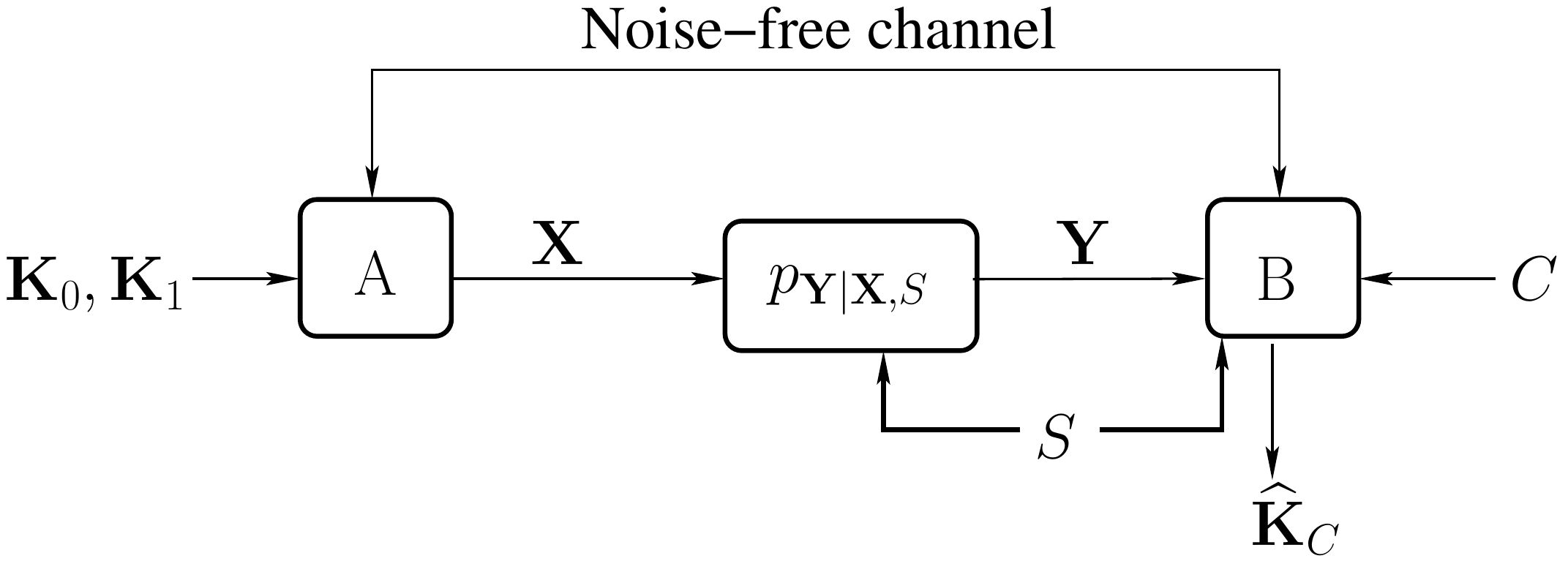}
\caption{Communication setup for oblivious transfer over channels with state}
\label{fig:otbasic}
\end{figure}

Communication under secrecy constraints has been studied
by many authors (see \cite{BlochB2011}). In particular, private communication over a wiretap channel
in the presence of eavesdropper has been studied 
extensively~\cite{Wyner1975,cheong,TanB2014,KhistiW2010,RezkiKA2014,bloch2008}.
In this work, we make use of coding techniques for
Gaussian wiretap channels as a building block for our achievability schemes.

In both OFDM and MIMO, we rely on the modeling of the channel
as parallel fading channels. For the MIMO setup, this is done
using the SVD precoder matrix that is communicated by Bob to Alice.
The parallel channels are grouped in pairs. 
OT is performed independently at different rates over different pairs.


We show (Theorem~\ref{thm1}) that the best pairing of the parallel channels
is that of the strongest channel with the weakest, and so on with the
rest of the channels.
The idea of pairing good and bad subchannels in OFDM and SVD-precoded
MIMO was also used in \cite{Mohammed2011july,Mohammed2011june}
with the aim of designing signal sets that minimize error probability or maximize mutual information.
Here, we exploit subchannel pairing to guarantee that Alice is oblivious
to which file is requested and that Bob only receives one of the two files.
We also derive the optimal power allocation among the pairs of channels.

The paper is organized as follows. Section~\ref{sec:model} presents the problem definition and 
the system model for both OFDM and MIMO channels. In Section~\ref{sec:scheme}, we present
protocols for OT over 2-channels OFDM, $2\times2$ MIMO and $2\times 1$ MIMO channels.
We present the general protocol for $2N$-channels OFDM and $2N \times n_B$ MIMO models
in Section~\ref{sec:protocol}, following a common principle.
Optimization of our protocol is discussed in Section~\ref{sec:optimize}.
High SNR asymptotics of OT rate for our protocol is analyzed in Section~\ref{sec:high}.
We provide simulation results of our OT scheme for simple OFDM and MIMO channels
in Section~\ref{sim:results}. Finally, we conclude the paper in Section~\ref{sec:conc}. The proof of our optimal pairing (Theorem~\ref{thm1}) is presented
in Appendix~\ref{sec:app1}.

\section{System Model}
\label{sec:model}
Alice (A) and Bob (B) are two parties in the system as shown 
in Fig.~\ref{fig:otbasic}. Alice has two 
binary strings $\bK_0, \bK_1$ of equal length, and Bob wants one of these 
strings $\bK_C$ where $C\in \{0,1\}$ is Bob's {\it choice bit}. 
We assume that all the bits in $(\bK_0,\bK_1,C)$ are
i.i.d. $\sim Ber(1/2)$. Alice can communicate with Bob over a channel
$p_{Y|X,S}$ with state $S$, where the state remains fixed over a large
block length $n$, and varies from block to block in an i.i.d. manner.
The state is known to Bob at the beginning of a block. This models wireless
communication setups, where in a large coherence block of length $n$,
the fading state remains fixed, and the fading state is known (estimated)
by the receiver. This is commonly known as the {\em quasi-static
channel model} \cite{Tse_book},\cite{BlochB2011}. 
In addition to this channel, there is also a noise-free channel
over which Alice and Bob can communicate real numbers
between each other without any error/distortion. 
During each block,
the noise-free channel can be used any finite number of times.
The length $L(S)$ of $\bK_0,\bK_1$ depends on $S$. Since Bob knows the state $S$
at the beginning of a block, he is assumed to compute and communicate
$L(S)$ to Alice over the noise-free channel.
The goal of a protocol is to transfer $\bK_C$ to Bob
obliviously, within the current block, such that Bob has negligible
knowledge about $\bK_{\overline{C}}$, and Alice has no knowledge
about $C$ ({\it perfect secrecy} against Alice). 

Our setup can also be used to transfer large files.
We then need multiple
coherence blocks to complete the OT session for one pair of files. The two
files can be broken into multiple chunks to form one pair $(\bK_{0i},\bK_{1i})$ 
for each block $i$. Then one run of the protocol is performed in each block,
where the choice bit $C$ of Bob remains the same over the whole session
involving many runs of the protocol.

An $(n,L(\cdot))$ OT protocol is parameterized by the number $n$ of channel
uses and by a function $L(\cdot)$ of the state $S$.
There are a total of $k$ rounds of communication between Alice and Bob, including communication
over both the noisy and noise-free channels. These are indexed by
$1,2,\cdots, k$, where $k$ can be random and can be dependent on $S$. But for every
$S$, it is required to be finite with probability $1$.
The noisy channel is used at rounds
$i_1,i_2,\cdots, i_n \in \{1,\cdots ,k\}$. At every round before round $i_1$, between consecutive $i_j$ and $i_{j+1}$, and after round $i_n$, Alice
and Bob exchange a sequence of real numbers over the noise-free channel.
In the following, $X_i$ and $Y_i$ denote respectively the input and the output 
of the noisy channel at time index $i$. 
In the following description of the protocol, we denote
$\bY^i:= (Y_1,Y_2,\cdots, Y_i)$ 
for any positive integer $i$. $\bE^i, \bF^i$ are also similarly defined.
In the rest of the paper, we also denote the transmitted length-$n$ vector
by $\bX$. The length-$n$ vector transmitted by the $l$-th antenna (in case
of MIMO) or over the $l$-th subchannel (in case of OFDM) will be denoted
by $\bX_l=(X_{l1}, X_{l2},\cdots, X_{ln})$.

\subsection{The structure of an $(n,L(\cdot))$ protocol:}

\begin{enumerate}
 \item Alice has two bit-strings $\bK_0,\bK_1$ of length $L(S)$ each, and Bob has a
choice bit $C$. $\bK_0,\bK_1$ can be substrings of two larger strings available
with Alice, and their length $L(S)$ is computed by Alice based on some 
information about $S$ sent by Bob during the protocol. 
\item Alice and Bob generate private random variables
 $W_A,W_B,$ respectively.
 \item For $i_j<i<i_{j+1}$ for every $j=0,1,\cdots,n$ (assuming $i_0=0$
and $i_{n+1}=k+1$),  
Alice sends
$E_i = E_i(\bK_0,\bK_1,W_A,\bF^{i-1})$ and Bob sends
$F_i = F_i(C,S,W_B,\bE^{i-1},\bY^{j})$
over the noise-free channel. Here $F^0=E^0=Y^0 = \emptyset.$
\item For $i=i_j$, Alice transmits $X_j = X_j(\bK_0,\bK_1, W_A, \bF^{i_j-1})$
over the noisy channel
and Bob receives $Y_j$.
There is no communication over
the noise-free channel in these rounds, and thus $E_i=F_i=\emptyset$.
\item At the end of the protocol,
Bob computes $\widehat{\bK}_C = \widehat{\bK}(C,S,W_B,\bE^k,\bY^n)$.
\end{enumerate}
\vspace{2mm}
The rate $L(S)/n$ of a protocol as described above is a function of the state
$S$, and is denoted by $R(S)$.
\begin{definition}
A non-negative rate function $R(S)$ is said to be achievable  if
there is a sequence of $(n,L^{(n)}(\cdot))$-protocols such that for every $S$, $\frac{L^{(n)}(S)}{n}
\rightarrow R(S)$ as $n\rightarrow \infty$, and the protocols
satisfy the conditions
\begin{eqnarray}
&P(\widehat{\bK}_C \neq \bK_C)\rightarrow 0 \notag\\
&I(\bK_0\bK_1W_A\bF^k;C) = 0 \notag\\
&\frac{1}{n} I(CSW_B\bY^n\bE^k; \bK_{\overline{C}}) \rightarrow 0. \label{eq:secrecyBob}
\end{eqnarray}
\end{definition}
The average rate $R$ is the expectation of $R(S)$.
The OT capacity is the supremum of all achievable average OT rates.

\subsection{Gaussian wiretap channel}

Wiretap channel has been studied as a standard model for communication
in the presence of an eavesdropper~\cite{Wyner1975,cheong}.
We model our MIMO and OFDM channels as complex channels.
If Alice and Bob are respectively
the transmitter and receiver of a complex AWGN channel, and if Eve is
a wiretapper, whose received symbol is more noisy than that of Bob (degraded
channel assumption),
then the {\em secrecy capacity} of the wiretapper channel is given by
\begin{align}
\Capac\left(\frac{P}{\sigma_B^2},\frac{P}{\sigma_E^2}\right)
= \log_2 \left(1+\frac{P}{\sigma_B^2}\right)
- \log_2 \left(1+\frac{P}{\sigma_E^2}\right)
\label{eq:gwiretap}
\end{align}
where $\sigma_B^2$ and $\sigma_E^2$ are the variance of the noise
at Bob and Eve, respectively, and $P$ is the transmit power~\cite{cheong}. Encoding for such channels
involves mixing the message with some random bits (with rate equaling
the capacity of the wiretapper) before encoding for the complex AWGN channels.
Bob can decode both the message and the random bits as the total rate
of these is below his capacity, whereas the random bits completely hide
the message from Eve.
Eve gets almost no information about the message~\cite{TanB2014}.
We will denote this channel with power constraint $P$ as
$\wiretap (P,\frac{P}{\sigma_B^2},\frac{P}{\sigma_E^2})$.
Practical coding schemes approaching the secrecy capacity have been proposed
for discrete memoryless channels using polar codes \cite{Mahdavifar2011}
and for the Gaussian channel based on lattice codes \cite{Ling2012},
under semantic security.

In this paper we consider two channels with states, OFDM and MIMO, as 
discussed below.
The essential technique used for OT over both these setups is the same.


\subsection{The OFDM Setup}
\label{sec:ofdm}
The OFDM setup is modeled in Fig.~\ref{fig:otmimo} as
$2N$ parallel fading AWGN channels between Alice and Bob.
The channel states are given by independent fading
coefficients $H_0,H_1,\cdots,H_{2N-1}$.
If the vector $\bX_l = (X_{l1},X_{l2},\cdots,X_{ln})$ is transmitted
in $n$ channel uses over the $l$-th channel for $l=0,1,\cdots,2N-1$,
then the received vector over the $l$-th channel is given by
\begin{align*}
& \bY_l = H_l\bX_l + \bZ_l,
\end{align*}
where $\bZ_l$ is the noise with  i.i.d. real and imaginary parts $\sim {\cal N}(0,1/2)$.
We assume that $H_l$ are i.i.d. with Rayleigh distribution.
The channel gains remain fixed for
a block of length $n$, and change from block to block in an i.i.d. manner.
We assume that they are known to Bob in the beginning of the block.
The average transmitted power in any block is restricted to $P$, i.e.,
$\sum_{l=0}^{2N-1} \sum_{j=1}^n |X_{lj}|^2 \leq nP$.

\vspace*{-2mm}
\begin{figure}[htbp]
\centering
\includegraphics[height=50mm,width=\columnwidth]{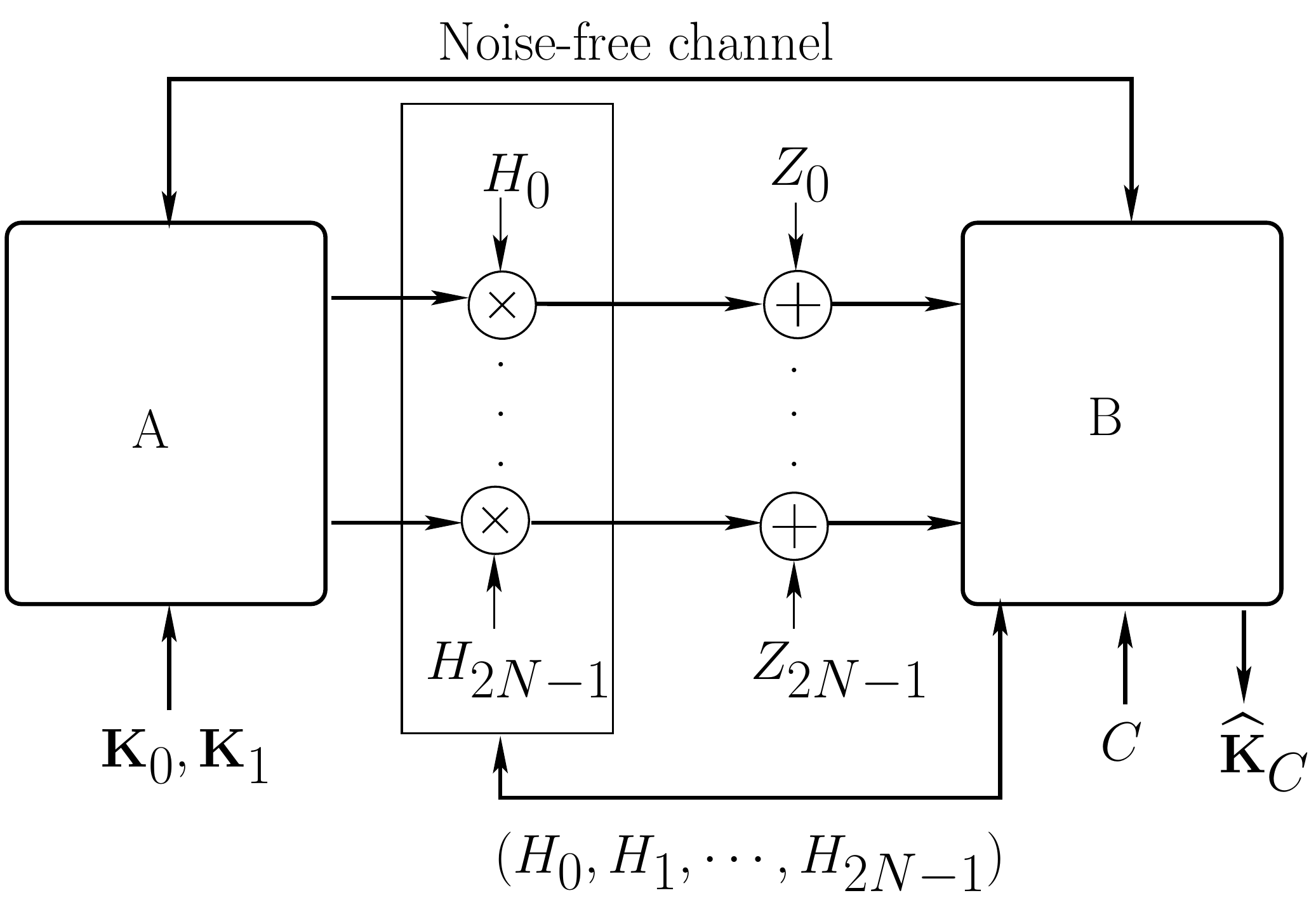}
\caption{The OT setup with independent parallel channels}
\label{fig:otmimo}
\end{figure}

\subsection{The MIMO Setup}
Let us consider the MIMO system with transmitter Alice and receiver Bob, as shown in Fig.~\ref{fig:mimo}.
The transmitter has $n_A$ antennas and
the receiver has $n_B$ antennas. We assume that $n_A$ is even.
Let $\bX = (X_{lj})_{\substack{ 0\leq l \leq n_A-1\\  1 \leq j \leq n}}$ denote the complex matrix transmitted by Alice over $n$ uses of the MIMO channel.
The received matrix $\bY$ is given by
\begin{align}
&\bY = \bH\bX +\bZ
\label{eq:mimoch}
\end{align}
where $\bZ \in \mathbb{C}^{n_B \times n} $ is the complex Gaussian noise matrix with
all entries having i.i.d. real and imaginary parts $\sim \bGa(0,1/2)$
and $\bH\in \mathbb{C}^{n_B\times n_A}$ represents the complex channel 
fading matrix. The entries of $\bH$ are assumed
to be i.i.d. complex random variables with independent real and imaginary 
parts $\sim \bGa(0,1/2)$. $\bH$ remains fixed
over the block of length $n$, and changes in an i.i.d. manner from block to
block. The average transmit power in any block is constrained to be $P$, i.e.,
$\sum_{l=0}^{n_A-1}\sum_{j=1}^n |X_{lj}|^2 \leq nP$.
We assume that $\bH$ is known only to Bob in the beginning of each block.

\begin{figure}[htbp]
\hspace*{7mm}\includegraphics[width=1\columnwidth]{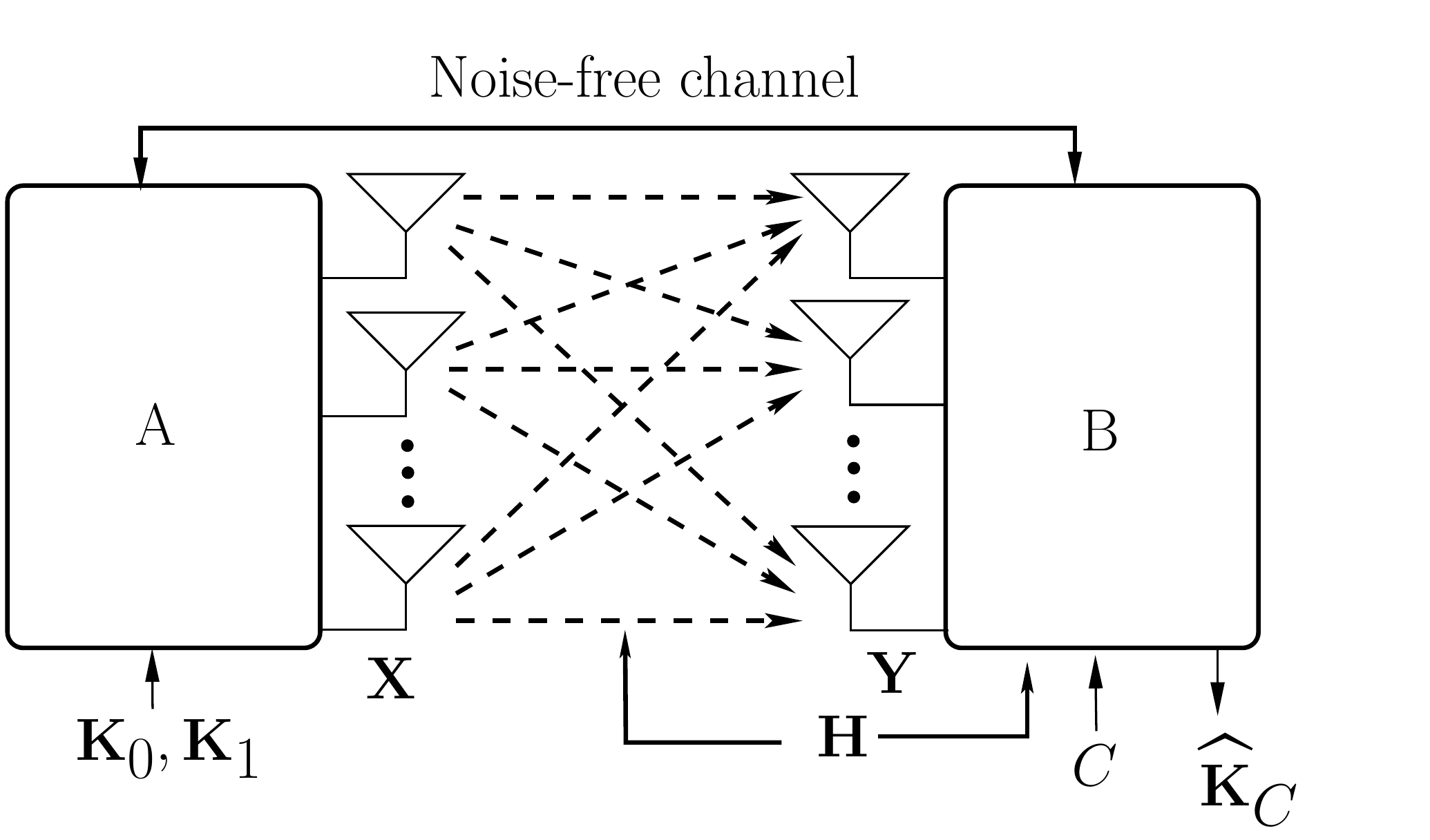}
\caption{MIMO system for oblivious transfer}
\label{fig:mimo}
\end{figure}

\section{The Protocol: Some Examples}
\label{sec:scheme}

We now show our OT protocols for some simple examples to illustrate
the basic principle. In all the three examples, Bob reveals some partial
information about the channel state to Alice so that there are, in effect, two
parallel channels with different SNRs, and Alice does not know which
of them is the better channel. Bob reveals the channel over which
each file is to be communicated -- the desired file over the stronger channel,
and the other file over the weaker channel. Alice uses encoding for a suitable 
wiretap channel so that Bob can decode the file transmitted over the
stronger channel, but not the file transmitted over the weaker channel.
\subsection{ $2$-Channels OFDM}
Let us consider an OFDM setup
with $2$ subchannels, each of which undergo independent and identical
Rayleigh fading.
For a block, let us define 
\begin{align*}
B & =\arg\max \{|H_0|,|H_1|\}\\
W & = C \oplus B \\
R & = \Capac(P|H_B|^2/2,P|H_{\overline B}|^2/2) - \epsilon
\end{align*}
where $\oplus$ denotes the modulo-$2$ addition, $\Capac(\cdot,\cdot)$
is given in \eqref{eq:gwiretap}, and $\epsilon>0$
is a pre-chosen constant. 

\vspace*{2mm}
\noindent
{\it The protocol:}
\begin{enumerate}
 \item Bob reveals $(W,|H_B|,|H_{\overline{B}}|)$ to Alice over the noise-free channel.
 \item Alice takes strings $\bK_0$ and $\bK_1$ of length $L(|H_0|,|H_1|):=nR$ each. She encodes
      $\bK_W$ and $\bK_{\overline{W}}$ into two length-$n$ codewords $\bX_0$ 
      and $\bX_1$ respectively, such that each has an average power $P/2$.
      A code suitable for $\wiretap (\frac{P}{2},\frac{P|H_B|^2}{2},\frac{P|H_{\overline{B}}|^2}{2})$ 
      is used to encode both the strings.
      $\bX_0$ and $\bX_1$ are transmitted over the respective channels.
      Note that $\bK_C$ has been encoded into $\bX_B$, and $\bK_{\overline{C}}$
      has been encoded into $\bX_{\overline{B}}$.
 \item Bob receives $\bY_0$ and $\bY_1$ with SNR $P|H_0|^2/2$ and $P|H_1|^2/2$
      respectively. He decodes $\bK_C$ from $\bY_B$ using the decoder for
      the wiretap channel referred above.

\end{enumerate}

\vspace*{2mm}
\noindent
{\it Correctness of the protocol:}
Note that $\bK_C$ is transmitted over the stronger channel ($B$), and
$\bK_{\overline{C}}$ is transmitted over the weaker channel ($\overline{B}$).
Bob's received SNR in the stronger channel is $P|H_B|^2/2$, whereas
his received SNR in the weaker channel is $P|H_{\overline B}|^2/2$.
Thus he can decode $\bK_C$ with vanishing probability of error, whereas
he can get negligible information about $\bK_{\overline C}$ as his SNR
is that of the wiretapper in this channel. Since $|H_0|$ and $|H_1|$ are
independent and identically distributed, it is easy to check that
$I(W;C)=0$, thus Alice does not learn anything about Bob's
choice $C$.

\delq{
\vspace*{2mm}
{\bf Generalization to $2N$-channels OFDM:}                  
The protocol for $2$-channels OFDM can be generalized to
$2N$-channels OFDM in a simple way. Bob pairs the channels
into $N$ pairs of channels.
Bob reveals these pairs to Alice.
He also reveals the set of fading coefficients in each pair 
without telling which one
is for which channel. Alice finds two $N$-tuples of
channels by taking the first of each pair of channels in one
$N$-tuple, and by taking the second of each pair in the other $N$-tuple.
Bob asks for the desired ($\bK_C$) file to be transmitted over the better 
of the channels in the pairs, and the other file over the weaker of
the channels. Based on the fading states of each pair,
Alice will allocate different amount of 
power between different pairs while respecting her average power constraint. 
In each pair of channels, exactly the same protocol as presented for
the $2$-channels OFDM is followed to transfer a part of the desired file
obliviously. 
It is worth noting that we have two choices to make: 
(i) the pairing of the channels,
and (ii) the power allocation over different pairs of channels.
We omit these details here due to limited space.
}
\subsection{$2\times 2$ MIMO}
Consider a $2\times 2$ fading MIMO channel between Alice and Bob.
Alice and Bob each has $2$ antennas. Let ${\bf H}$ denote the
$2\times 2$ complex fading matrix. The input-output relation
for the channel is given by \eqref{eq:mimoch},
where $\bY, \bX,\bZ$ are $2\times n$ matrices.

Let the SVD decomposition of ${\bf H}$ be given by
\begin{align*}
&{\bf H} = {\bf U} \boldsymbol{\Lambda} {\bf V}^H, 
\end{align*}
where $\boldsymbol{\Lambda}$ is a diagonal matrix with diagonal elements ${\lambda}_0,
\lambda_1$ such that $\lambda_0\geq \lambda_1$. These are the (real) singular
values of $\bH$.
Let $\bV_0,\bV_1$ denote the columns of ${\bf V}$.
We define 
\begin{align}
& (\bW_0,\bW_1) = (\bV_C,\bV_{\overline C})\notag\\
\mbox{and } & R = \Capac(P\lambda_0^2/2,P\lambda_1^2/2) - \epsilon
\label{eq:twobytwo2}
\end{align}
for some pre-decided $\epsilon$, where the $\Capac(\cdot,\cdot )$ above is defined
in \eqref{eq:gwiretap}. 
Note that $\bW_0,\bW_1$ are the same as $\bV_0,\bV_1$,
but permuted depending on $C$. Bob shares $(\bW_0,\bW_1)$ with Alice
in our protocol, and Alice uses it as the precoding matrix.
Bob first pre-multiplies the received matrix by $\bU^H$.
The resulting end-to-end system is shown in Fig.~\ref{fig:mimo2by2} where a switch,
controlled by Bob's choice bit $C$, determines which input of Alice
passes through which channel to Bob. The firm lines and dotted lines
show the two positions of the coupled switch.


\vspace*{2mm}
\noindent
{\it The protocol:}
\begin{enumerate}
 \item  Bob reveals $(\bW_0,\bW_1,\lambda_0,\lambda_1)$ to Alice over the 
       noise-free channel.
 \item The basic transmitter and receiver block diagram is shown in Fig.~\ref{fig:mimoot}.
      Alice computes $R$ using \eqref{eq:twobytwo2}, and takes strings $\bK_0$ and $\bK_1$ of length $L(\lambda_0,\lambda_1):=nR$ each. She encodes
      $\bK_0$ and $\bK_1$ into two length-$n$ codewords $\bX_0$
      and $\bX_1$ respectively, such that each has an average power $P/2$.
      A code suitable for $\wiretap (\frac{P}{2},\frac{P\lambda_0^2}{2},\frac{P\lambda_1^2}{2})$ 
      is used to encode both the strings.
      She then transmits the matrix
      \begin{align*}
      \left[\bW_0\,\,\, \bW_1\right]
      \left[\begin{array}{l} \bX_0\\ \bX_1\end{array}\right]
      & =\bW_0 \bX_0 + \bW_1\bX_1 \\
      & = \bV_0 \bX_C + \bV_1\bX_{\overline C} \\
      & =\bV\left[\begin{array}{l} \bX_C\\ \bX_{\overline C}\end{array}\right].
      \end{align*}
 \item Bob first multiplies the received $2\times n$ matrix by $\bU^H$.
      The resulting end-to-end channel is given by
      \begin{align}
      {\widetilde\bY} = \left[\begin{array}{l} \widetilde{\bY}_0 \\ \widetilde{\bY}_1   \end{array} \right] & = \bU^H \bH\bV\left[\begin{array}{l} \bX_C\\ \bX_{\overline C}\end{array}\right] +\bU^H\left[ \begin{array}{l} \bZ_0\\ \bZ_1  \end{array} \right] \notag\\
	& = \left[\begin{array}{l} \lambda_0 \bX_C\\ \lambda_1 \bX_{\overline C}\end{array}\right]
      + \bU^H\left[ \begin{array}{l} \bZ_0\\ \bZ_1  \end{array} \right]. \label{eq:rec2by2}
      \end{align}
      Bob gets $\widetilde{\bY}_0$ and $\widetilde{\bY}_1$ with SNR $P\lambda_0^2/2$ 
      and $P\lambda_1^2/2$ respectively. He decodes $\bK_C$ from $\bY_0$ using the 
      decoder for the wiretap channel referred above.
\end{enumerate}

\begin{figure}[htbp]
\centering
\includegraphics[width=\columnwidth]{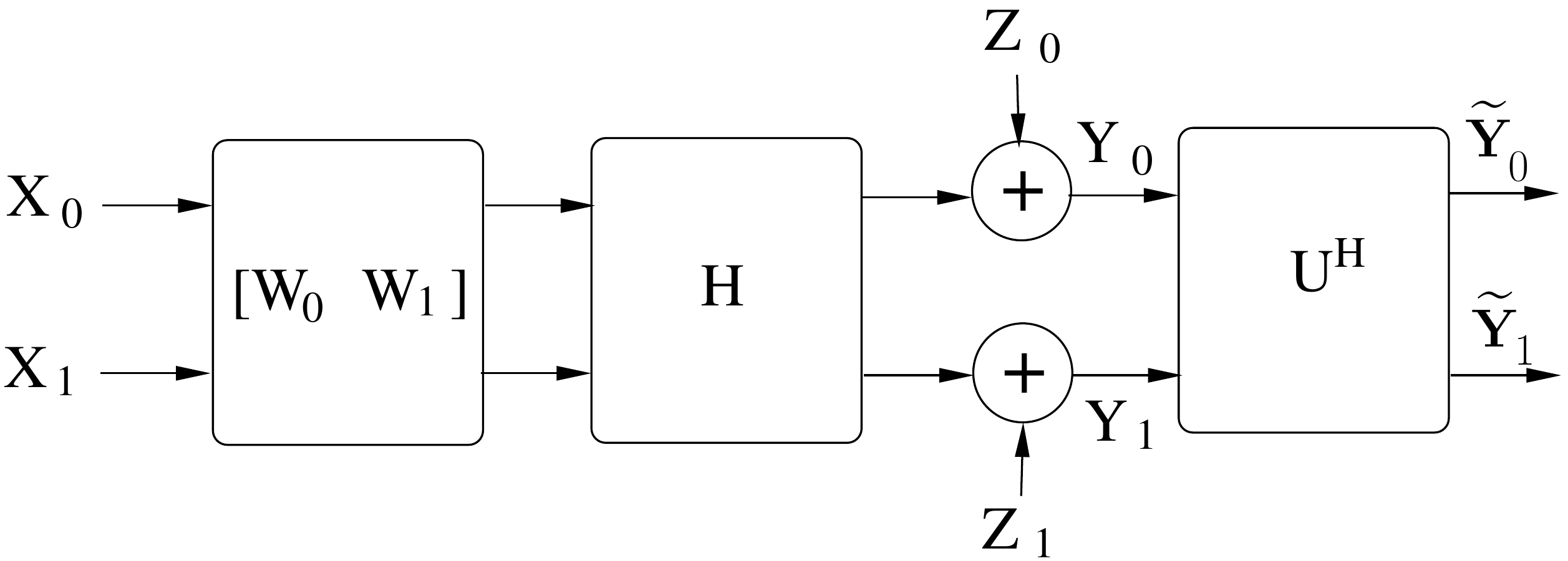}
\caption{MIMO precoding for OT}
\label{fig:mimoot}
\end{figure}

\begin{center}
\begin{figure}[h]
\includegraphics[width=0.9\columnwidth]{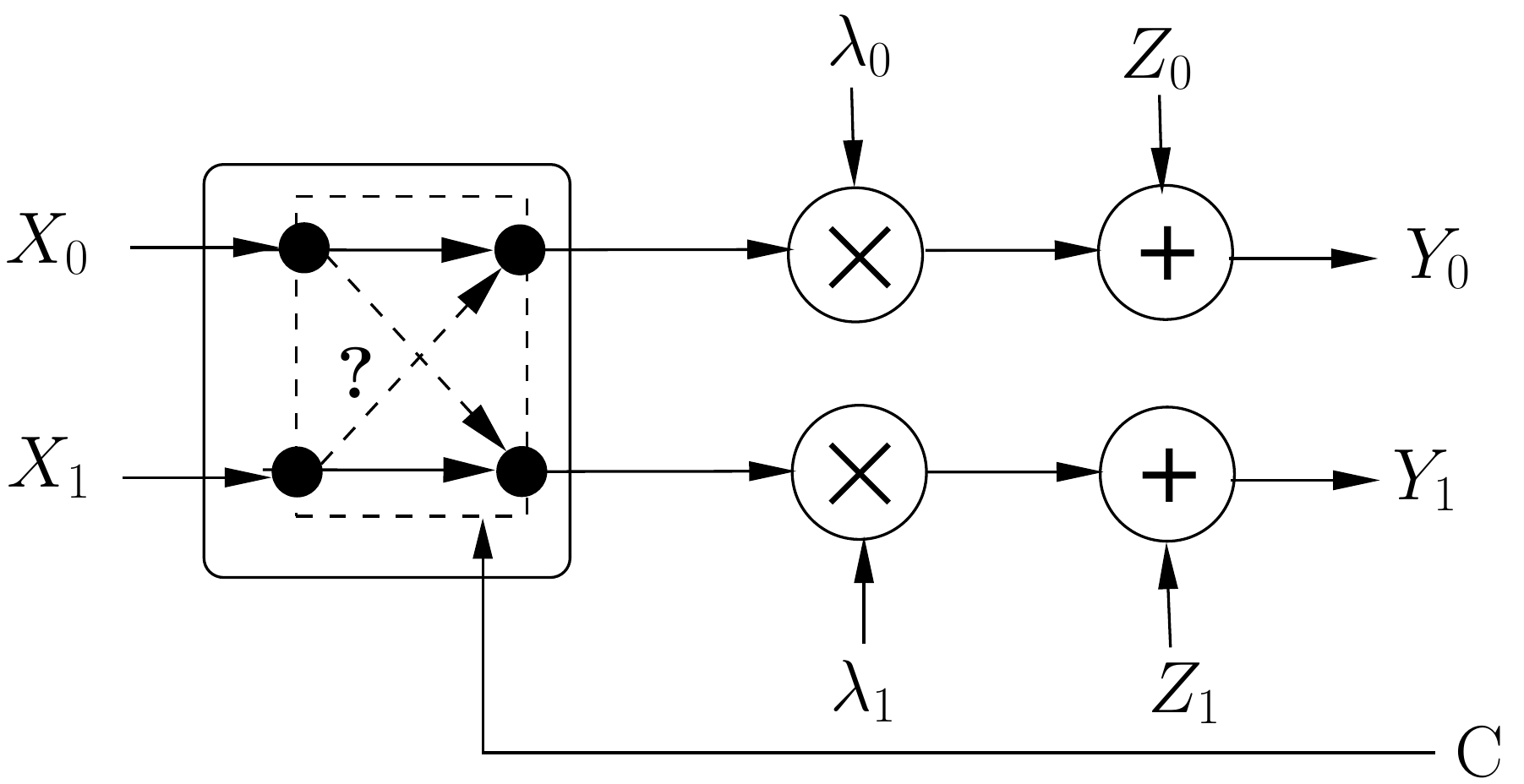}
\caption{The equivalent channel with a switch for $2\times 2$ MIMO setup}
\label{fig:mimo2by2}
\end{figure}
\end{center}

\noindent
{\it Correctness of the protocol:}
First note that since $\widetilde{\bY}$ is obtained by a unitary 
(hence invertible) transformation on $\bY$, it contains exactly the same 
information as $\bY$. So we will henceforth treat $\widetilde{\bY}$ as Bob's 
received matrix.
Since $\bU$ is a unitary matrix, $\bU^H\bZ$ has the same distribution as that
of $\bZ$. Also note that $\bK_C$ is encoded into $\bX_C$, which is received
as ${\widetilde\bY}_0$ with SNR $P\lambda_0^2/2$. 
Since this encoding is done by Alice for a complex Gaussian wiretap channel 
with the same receiver
SNR, Bob can decode $\bK_C$ with vanishing probability of error.
On the other hand, $\bK_{\overline{C}}$ is encoded into $\bX_{\overline{C}}$, 
which is received as ${\widetilde\bY}_1$ with SNR $P\lambda_1^2/2$.
Bob can get negligible information
about $\bK_{\overline C}$ as his SNR in ${\widetilde\bY}_1$ is that of the wiretapper.
This ensures secrecy of Alice against Bob. 

About the secrecy of Bob against Alice, first note that $\bH$
is circularly symmetric, and thus $(\bV_0,\bV_1)$
and $(\bV_1,\bV_0)$ have the same distribution, that is, their joint distribution
is symmetric in $\bV_0$ and $\bV_1$. Also, note that $\lambda_0,\lambda_1$
are independent of $C,\bV_0,\bV_1$. Thus
\begin{align*}
I(\bW_0,\bW_1,\lambda_0,\lambda_1;C) 
& = I(\bV_C,\bV_{\overline{C}};C) = 0. 
\end{align*}
This ensures the secrecy of Bob against Alice.

As seen in \eqref{eq:rec2by2}, the SVD precoding as shown in Fig.~\ref{fig:mimoot} transforms the
MIMO channel into a parallel fading Gaussian channel, where Alice
is unsure of which of the two channels has the gain $\lambda_0$,
and which has gain $\lambda_1$. 
We now discuss the $2\times 1$ MIMO
system where the same technique takes a simple elegant form.

\subsection{ $2\times 1$ MIMO}
Consider a $2\times 1$ fading MIMO channel between Alice and Bob. 
Let ${\bf H}=(H_0,H_1)$ denote the
$1\times 2$ fading matrix such that the symbol received by Bob
over the MIMO channel is given by
\begin{align*}
&Y={\bf H}{\bf X}+Z,
\end{align*}
where ${\bf X}=(X_0,X_1)^T$ is the vector transmitted by Alice,
and $Z$ is the noise. Over $n$ uses of the channel,
the received vector is given by 
\begin{align*}
&\bY={\bf H}{\bf X}+\bZ,
\end{align*}
where ${\bf X}$ and $\bZ$ are respectively the $2 \times n$ transmitted matrix
and the noise vector of length $n$.
Let the SVD of 
${\bf H}$ be
\begin{align*}
&{\bf H} = \bLambda {\bf V}^H 
\end{align*}
where $\bLambda=(\lambda, 0)$, $\lambda = \sqrt{|H_0|^2+|H_1|^2}$,
the first column of ${\bf V}$ is ${\bf V}_0=(1/\lambda ){\bf H}^H$, and the second 
column of ${\bf V}$ is a unit vector ${\bf V}_1$ orthogonal to ${\bf H}$.

The best way to communicate messages (without any secrecy condition) is 
using SVD precoding 
wherein Alice multiplies her message symbol with the first column of
${\bf V}_0$ and transmits. Bob simply divides the received symbol by $\lambda$
and chooses the message symbol nearest to the result. Note that
if in addition, Alice added any scalar multiple of
${\bf V}_1$ to her transmission, it would not contribute
to the received symbol as ${\bf V}_1$ is orthogonal to
$\bH$. Thus this dimension which is orthonormal to ${\bf H}$ (the null-space
of ${\bf H}$) is not useful for communication, as it has zero
gain. This reduces the MIMO channel to a single fading AWGN
channel with fading coefficient $\lambda$.

We now give an OT protocol for this channel when only Bob has
the knowledge of ${\bf H}$ at the beginning of a block. 
We define
\begin{align}
& (\bW_0,\bW_1) = (\bV_C,\bV_{\overline C}) \label{eq:twobyone1}\\
 \mbox{and } & R = \log_2 \left(1+\frac{P\lambda^2}{2}\right) 
                  - \epsilon
\label{eq:twobyone2}
\end{align}
for some pre-decided $\epsilon$. Bob shares $(\bW_0,\bW_1)$ with Alice
in our protocol, and Alice uses it as the precoding matrix. 
The resulting channel is equivalent to what is shown in Fig.~\ref{fig:mimo2by1} where a switch,
controlled by Bob's choice bit $C$, determines which input of Alice
passes through the channel to Bob.


\vspace*{2mm}
\noindent
{\it The protocol}
\begin{enumerate}
 \item Bob reveals $(\bW_0,\bW_1,\lambda)$ to Alice over the noise-free channel.
      He sets $(\bW_0,\bW_1)$ as in \eqref{eq:twobyone1}.
 \item Both Alice and Bob compute $L(\lambda): = Rn$ with $R$ given in
       \eqref{eq:twobyone2}.
      Alice encodes each of $\bK_0$ and $\bK_1$
      (of length $L(\lambda)$ each) into a $n$-length vector.
      She uses a code suitable for a complex AWGN channel with SNR $\frac{P}{2}\lambda^2$.
      Let these encoded vectors be $\bX_0$ and $\bX_1$ respectively. 
      Over $n$ uses of the channel, Alice transmits the $2\times n$ matrix
      ${\bf W}_0\bX_0+{\bf W}_1\bX_1$.
 \item Bob receives
      \begin{align*}
      \bY & = \bH({\bf W}_0\bX_0+{\bf W}_1\bX_1) +\bZ\\
	& = \lambda \bX_C + \bZ. 
      \end{align*}
      Bob now decodes $\bK_C$ from $\bY$ with probability of error
      going to zero as $n \rightarrow \infty$.
\end{enumerate}

\begin{center}
\begin{figure}[h]
\includegraphics[width=0.9\columnwidth]{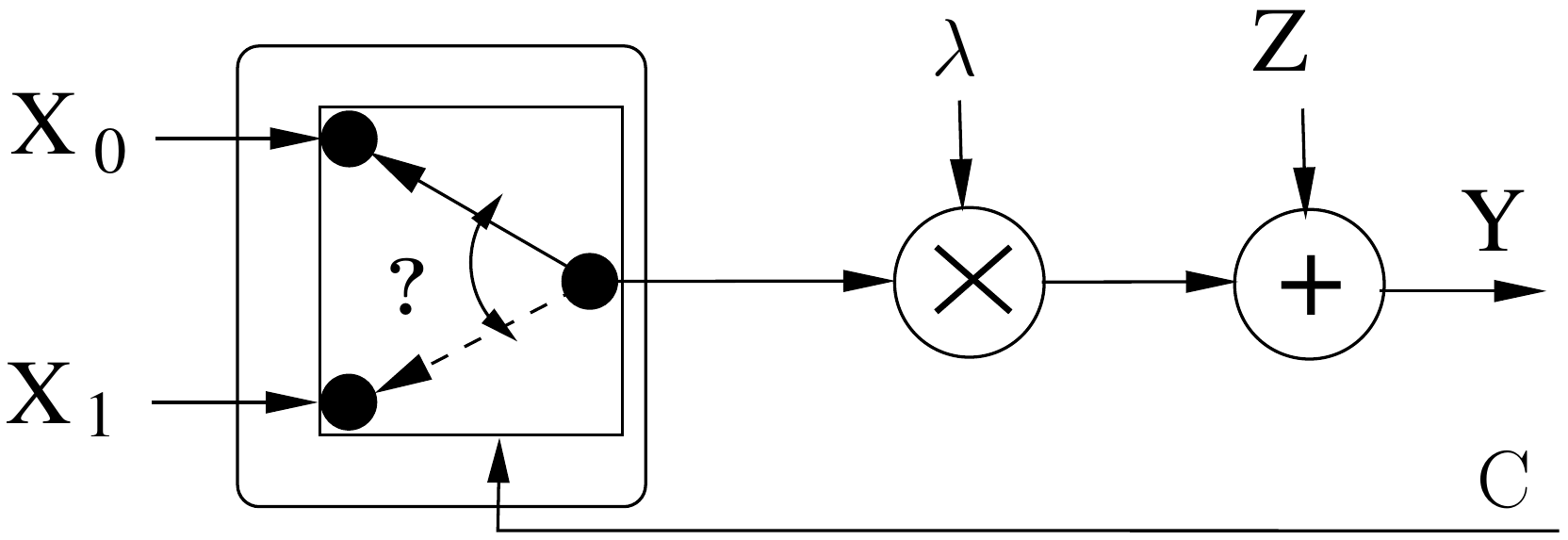}
\caption{The equivalent channel with a switch for $2\times 1$ MIMO setup}
\label{fig:mimo2by1}
\end{figure}
\end{center}

\vspace*{2mm}
\noindent
{\it Correctness of the protocol:}
Since $\bX_{\overline{C}}$ is transmitted
in the null-space of ${\bf H}$, it does not contribute to Bob's
received vector. Thus Bob has no information about $\bK_{\overline{C}}$.
Since $\bH$ has i.i.d. Gaussian entries, $(\bV_0,\bV_1)$
has a distribution which is symmetric in $\bV_0$ and $\bV_1$, and $\lambda$
is independent of $(\bV_0,\bV_1)$. Thus, 
$I(\bW_0,\bW_1,\lambda;C)=0$. Thus the secrecy of Bob against Alice
is met.

\delq{
\vspace*{2mm}
{\bf General MIMO channel:}
Our MIMO protocol can be extended to arbitrary even number
of transmit antennas. Bob will compute the SVD of $\bH$, and then
arrange the resulting parallel channels in pairs. The rest of
the protocol proceeds as discussed in the general OFDM case.
}

\section{The General Protocol}
\label{sec:protocol}
\def\bP{{\bf P}}
\def\bU{{\bf U}}
\def\bV{{\bf V}}
\def\bW{{\bf W}}

In this section, we present a protocol for the general $2N$-channels OFDM and
$2N\times n_B$-MIMO models. Here we assume that Alice has more ($2N$) antennas
than Bob has ($n_B$). The case $n_B >2N$ is similar, and is discussed
briefly later.

For the MIMO setup, we first discuss how Bob can reveal some
partial information about the channel matrix to reduce the channel to
a parallel channel. We will then treat both OFDM and MIMO models as
parallel channels and present a common OT protocol. The OT protocol
will group the parallel channels into pairs and perform OT over each
pair using similar technique as in the previous section.

\subsection{Reducing MIMO setup to parallel channels}
\label{subsec:mimo}
Let the SVD decomposition of ${\bf H}$ be given by
\begin{align*}
&{\bf H} = {\bf U} \boldsymbol{\Lambda} {\bf V}^H,
\end{align*}
where $\boldsymbol{\Lambda}$ is a $n_B\times 2N$ diagonal matrix with diagonal 
elements $\lambda_0 \geq \lambda_1 \geq \lambda_2 \geq \cdots \geq 
\lambda_{n_B-1}$. 
Let $\bP$ be a random $2N\times 2N$ permutation matrix chosen by Bob. 
Note that a permutation matrix is unitary, and thus $\bP^T = \bP^{-1}$.
Let us add $(2N-n_B)$ zero rows with $\bU^H$ to define the $2N\times n_B$ matrix
\begin{align*}
&{\widetilde{\bf U}} = \left[\begin{array}{c} {\bf U}^H \\ {\bf 0} \end{array}\right].
\end{align*}
Bob sends $\bW=\bV\bP$ over the noise-free channel, and Alice uses it as
the precoding matrix to transmit $\bV\bP \bX$. Bob first multiplies the
received vector $\bY$ by $\bP^T{\widetilde{\bf U}}$ to get
\begin{align*}
{\widetilde \bY} & = \bP^T {\widetilde{\bf U}} \bY\\
& = \bP^T\left[\begin{array}{c} \boldsymbol{\Lambda} \bP \bX + \bU^H\bZ \\ {\bf 0}\end{array}\right]\\
& = \bP^T\left[\begin{array}{c} \boldsymbol{\Lambda} \\ {\bf 0}\end{array}\right] \bP \bX + \bP^T\left[\begin{array}{c} \bU^H\bZ \\ {\bf 0}\end{array}\right]\\
\end{align*}
Let us denote $\underline{\blambda}:= 
(\lambda_0,\lambda_1,\cdots, \lambda_{2N-1})^T$ 
as the $2N$ length vector of diagonal elements of
$\left[\begin{array}{c} \boldsymbol{\Lambda} \\ {\bf 0}\end{array}\right]$
where $\lambda_l = 0$ for $l\geq n_B$. Let us also denote 
${\widetilde{\bZ}}:=\left[\begin{array}{c} \bU^H\bZ \\ {\bf 0}\end{array}\right]$.
Let $\pi$
denote the permutation induced on a vector by pre-multiplication by $\bP^T$,
that is, $\bP^T\underline{\blambda} = (\lambda_{\pi(0)},\lambda_{\pi(1)},
\cdots,\lambda_{\pi(2N-1)})$ in particular.
Then
\begin{align*}
{\widetilde Y}_l
& = \lambda_{\pi(l)} X_l + {\widetilde Z}_{\pi(l)}.\\
\end{align*}
We note that for $\pi(l) \geq n_B$, $\lambda_{\pi(l)}={\widetilde Z}_{\pi(l)}=0$.
This gives a set of parallel channels such that $2N-n_B$ of them
have zero gain and zero noise. These channels are completely
useless for communication.
Since $\bU^H$ is unitary, $\bU^H\bZ$
is also i.i.d. with independent real and imaginary components 
$\sim {\cal N}(0,1/2)$. Since Bob knows $\bP$ (and so $\pi$), 
he will neglect the channels $l$ for which $\pi(l) \geq n_B$.
To reduce this model to a standard parallel AWGN channels model with
constant noise variance in all channels but different channel gains,
we assume that Bob adds some independent noise with real and imaginary
parts $\sim {\cal N}(0,1/2)$ to each of
the channels for which $\pi(l) \geq n_B$.

We now prove a lemma which states that in the resulting
parallel channels, Alice can not know the order of the channel
gains.
\begin{lemma}
\label{lem:perm_inv}
Let $\bH$ be the channel matrix and $\bP$ is a permutation matrix chosen 
uniformly at random. Let $\bW=\bV\bP$ denote the precoding
matrix sent to Alice by Bob, and $\underline{\blambda}$ be the
zero-padded vector of ordered singular values.
Then for any $\bW$ and $\underline{\blambda}$, and for any two permutations 
$\bP$ and $\bP'$, we have $Pr(\bP|\bW,\underline{\blambda}) 
= Pr(\bP'|\bW,\underline{\blambda}) = \frac{1}{(2N)!}$.
\end{lemma}

\begin{proof}
$\bV$ is uniformly distributed over the set of $2N\times 2N$ unitary matrices
(see~\cite[Lemma~5]{telatar1999}).
Since $\bP$ is a unitary matrix $\bW=\bV\bP$ is also unitary and 
both $\bV\bP$ and $\bV\bP'$ are Haar matrices
with the same uniform distribution over the set of $2N\times 2N$ unitary matrices.
Hence $f_{\bW, \underline{\blambda}|\bP}(\bW,\underline{\blambda}|\bP)=f_{\bV, \underline{\blambda}}(\bW\bP^{T},\underline{\blambda})=f_{\bV, \underline{\blambda}}(\bW,\underline{\blambda})$, and also $f_{\bW, \underline{\blambda}}(\bW,\underline{\blambda})
=f_{\bV, \underline{\blambda}}(\bW,\underline{\blambda})$. So we have
$Pr(\bP|\bW,\underline{\blambda}) = \frac{1}{(2N)!}$.
\end{proof}

We have now reduced the MIMO channel to a standard parallel AWGN channels
with different gains (singular values) in different subchannels.
The above lemma says that from the partial channel state information given to Alice, she still would be 
 `completely uncertain' about the association of the singular values to the resulting subchannels.
 
\noindent
\underline{\it The case of $n_B >2N$:} When $n_B >2N$, 
$\bU$ is an $n_B\times n_B$ matrix and 
$\boldsymbol{\Lambda}$ is a $n_B\times 2N$ diagonal matrix
with $(n_B-2N)$ zero rows. Let the last $n_B-2N$ rows of $\bU^H, \bLambda$ and 
$\bU^H\bZ$ be removed to obtain respectively $\widetilde{\bU}, \widetilde{\bLambda}$ and $\widetilde{\bZ}$. 
As before, Alice transmits $\bV\bP \bX$. 
Bob first multiplies $\bP^T\widetilde{\bU}$ to the
received vector to obtain
\begin{align*}
 \widetilde{\bY} &= \bP^T \widetilde{\bU}\bY \\
& = \bP^T\widetilde{\bLambda}\bP\bX + \bP^T\widetilde{\bZ}.
\end{align*}
The protocol now continues with the $2N$ components of ${\widetilde{\bY}}$ which
constitute the output of the $2N$ parallel channels as before.

In the following, we consider a set of parallel channels indexed
by $1,2,\cdots, 2N$, as depicted in Fig.~\ref{fig:otmimo}. 
Such a model could have resulted from an OFDM channel
or a MIMO channel under the scheme discussed above. 
To treat MIMO and OFDM in a unified manner in the following, we also assume
$\lambda_l = |H_l|$ to be the channel gains in case of OFDM as they
provide the same performance. For OFDM, we assume that $\lambda_1,\lambda_2,\cdots,\lambda_{2N}$ are i.i.d. and Rayleigh distributed. We now define
an OT-pairing of the channels and a power allocation under a 
given total power constraint.
\begin{definition}
An OT-pairing of the $2N$ channels is defined using two maps 
$\ell,k:\{1,2,\cdots, N\} \rightarrow \{1,2,\cdots,2N\}$ 
such that 
\begin{enumerate}
 \item $\ell,k$ are $1-1$
 \item $Im(\ell)\cap Im(k)= \emptyset $
 \item $\lambda_{\ell(l)} > \lambda_{k(l)} \: \forall \: l $.
\end{enumerate}
The ordered pairs of the channels are then $(\ell(l),k(l)); l=1,2,\cdots, N$.
\end{definition}

\subsection{ Power allocation}
Alice divides the total average transmit
power $P$ between the subchannels. In our OT protocol, Alice transmits
the same power over the subchannels in a pair. Let $P_l$ the average power 
transmitted on each of the subchannels in pair $l$, that is, in
the subchannels $\ell(l)$ and $k(l)$, be $P_l$. Then $P_l \geq 0$ and
\begin{align}
& \sum_{l=1}^{N} P_l \leq \frac{P}{2}.
\end{align}

The rates for the pairs are taken as
\begin{align}
 R_l & = \Capac(P_l\lambda_{\ell(l)}^2,P_l\lambda_{k(l)}^2) -\epsilon
\label{eq:genrates}
\end{align}
for an arbitrarily small fixed constant $\epsilon >0$. We denote
$\underline{{\bf R}} = (R_1,R_2,\cdots, R_N)$.
Note that $R_l$ is close to the capacity of the wiretap channel 
$\wiretap (P_l, P_l\lambda_{\ell(l)}^2,P_l\lambda_{k(l)}^2)$.
Our OT protocol for the $2$-channels OFDM can be used with average
power constraint $2P_l$ to achieve a rate $R_l$ for each pair of subchannels.
The total rate achieved is thus
\begin{align}
& R = \sum_{l=1}^N \Capac(P_l\lambda_{\ell(l)}^2,P_l\lambda_{k(l)}^2) -\epsilon N.
\end{align}
For simplicity, we assume that $nR_l$ is an integer for each $l$.

We define for $l=1,2,\cdots, N$,
\begin{align}
\tilde{\gamma}_l & = (\gamma_{l0}, \gamma_{l1}) = (\ell(l),k(l)) \\
\tilde{\lambda}_l & = (\lambda_{\ell(l)},\lambda_{k(l)}), \label{eq:genchannel}
\end{align}
and denote $\underline{\tilde{\bgamma}}:=(\tilde{\gamma}_1,\tilde{\gamma}_2,\cdots,\tilde{\gamma}_N)$ and $\underline{\tilde{\blambda}}
=(\tilde{\lambda}_1,\tilde{\lambda}_2,\cdots, \tilde{\lambda}_N)$.

Let ${\bf T}$ denote the $2N\times 2N$ permutation matrix representing the
transposition of consecutive pairs. ${\bf T}$ consists of $N$ diagonal $2\times 2$
blocks $\left[\begin{array}{cc} 0 & 1\\ 1 & 0\end{array}\right]$.
We define
\begin{align}
\underline{\bgamma} & = \begin{cases} \underline{\tilde{\bgamma}} & \mbox{ if } C=0\\
 \underline{\tilde{\bgamma}} T & \mbox{ if } C=1 \end{cases} \label{eq:genpairs}
\end{align}
Bob shares $(\underline{\bgamma},\underline{\tilde{\blambda}})$ with Alice.
From Alice's point of view, the parallel channels appear to be associated
with the gains shown in Fig.~\ref{fig:mimogen}. The association of the gains to the
channels has one bit of uncertainty as depicted by the two possible positions
of the coupled switches. The position of the switches is controlled by $C$,
and is not known to Alice. We give the protocol below.
\begin{center}
\begin{figure}[h]
\includegraphics[width=0.9\columnwidth]{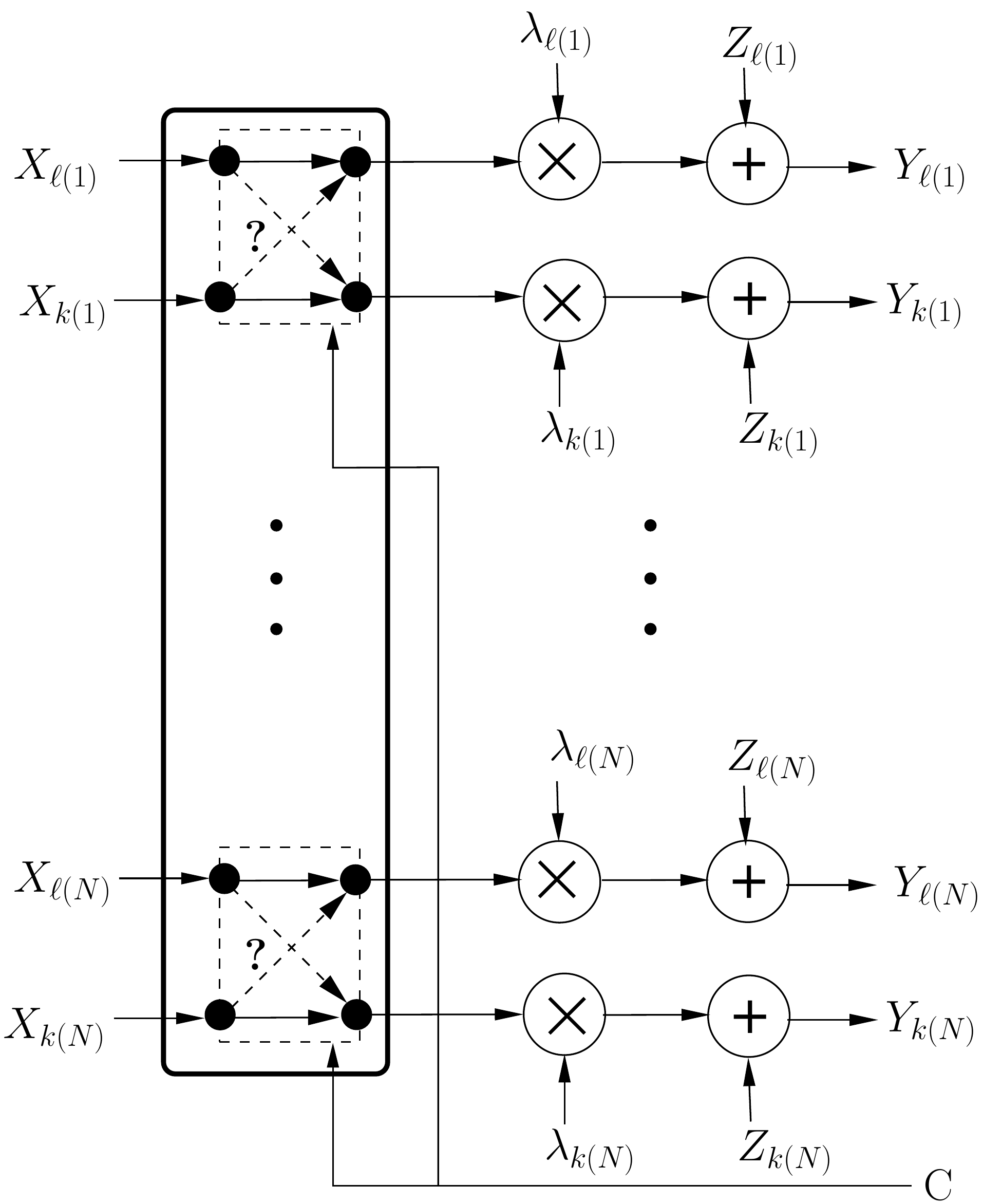}
\caption{The equivalent channel with a switch}
\label{fig:mimogen}
\end{figure}
\end{center}
\subsection{The protocol}
\begin{enumerate}
\item In case of a MIMO setup, Bob first reveals $\bW$ to Alice, and Alice
uses it as the precoding matrix. Bob also does appropriate pre-processing
as discussed in Sec.~\ref{subsec:mimo} to reduce the channel to a set of parallel channels.
 \item Bob selects an OT pairing $\ell,k$ and reveals 
	$(\underline{\bgamma},\underline{\tilde{\blambda}})$ to 
	Alice over the noise-free channel. He computes these using
	\eqref{eq:genpairs} and \eqref{eq:genchannel}
	respectively.
 \item Both Alice and Bob compute $R_l$ using \eqref{eq:genrates} and $L_l = R_ln$ 
      for $l=1,2,\cdots,N$. Let us denote $L = \sum_{l=1}^N L_l$.
      For each $j=0,1$, Alice breaks $\bK_j$ (of length $L$) into $N$ substrings
      $\bK_{jl};l=1,2,\cdots,N$ of lengths $L_l$ respectively. 
      For each $j=0,1,$ and $l=1,2,\cdots, N$, she encodes $\bK_{jl}$ into a 
      $n$-length vector $\bX_{jl}$ of average power
      $P_l$ using a code for the wiretap channel $\wiretap (P_l,P_l\lambda_{\ell(l)}^2,P_l\lambda_{k(l)}^2)$.
      Alice transmits this vector over $n$ uses of the channel $\gamma_{lj}$.
 \item Note that from \eqref{eq:genpairs}, $\gamma_{lC} = \ell(l)$ and $\gamma_{l\overline{C}}
      =k(l)$ for each $l=1,2,\cdots,N$. Thus Bob receives
      \begin{align*}
      \bY_{\ell(l)} & = \lambda_{\ell(l)}\bX_{Cl} +\bZ_{\ell(l)}.
      \end{align*}
      Bob now decodes $\bK_{Cl}$ from $\bY_{\ell(l)}$ with probability of error
      going to zero as $n \rightarrow \infty$.
      \end{enumerate}

\vspace*{2mm}
\noindent
{\it Correctness of the protocol:} 
Bob can decode $\bK_{Cl}$ from $\bY_{\ell(l)}$ for each $l$ with arbitrarily
small probability of error. This follows from standard results in
Gaussian wiretap channels~\cite{cheong}. It also follows that he gets only
an arbitrarily small amount of information about $\bK_{\overline{C}}$
from $\bY_{k(l)}$ in the sense of \eqref{eq:secrecyBob} \cite{TanB2014}. 

Alice knows that $\underline{\tilde{\bgamma}} \in \{\underline{\bgamma},\underline{\bgamma}T\}$.
Since $\underline{\bgamma}$ and $\underline{\tilde{\blambda}}$ are revealed to
Alice during the protocol, the uncertainty in $C$ is equivalent to
the uncertainty in which of $\underline{\bgamma},\underline{\bgamma}T$
is the value of $\underline{\tilde{\bgamma}}$. 

Now, let us first consider an OFDM channel. 
From the point of view of Alice,
\begin{align}
Pr(C=0|\underline{\bgamma}, \underline{\tilde{\blambda}})
& = Pr(\underline{\tilde{\bgamma}}=\underline{\bgamma}|\underline{\tilde{\bgamma}} \in \{\underline{\bgamma},\underline{\bgamma}T\}, \underline{\tilde{\blambda}})\notag\\
& = Pr(\underline{\tilde{\bgamma}}=\underline{\bgamma}T|\underline{\tilde{\bgamma}} \in \{\underline{\bgamma},\underline{\bgamma}T\}, \underline{\tilde{\blambda}})\label{eq:ofdm1}\\
& = Pr(C=1|\underline{\bgamma}, \underline{\tilde{\blambda}}).\notag
\end{align}
Here \eqref{eq:ofdm1} follows as we have assumed that the channel
gains of the parallel channels are i.i.d.
This implies that $I(C;\underline{\bgamma}, \underline{\tilde{\blambda}})=0$.

Similarly, if the parallel channels have resulted from a MIMO channel,
then Alice has also learned the precoding matrix $\bW$. Now,
\begin{align}
Pr(C=0|\bW, \underline{\bgamma}, \underline{\tilde{\blambda}})
\hspace*{-15mm} &\notag \\
& = Pr(\underline{\tilde{\bgamma}}=\underline{\bgamma}|\bW, \underline{\tilde{\bgamma}} \in \{\underline{\bgamma},\underline{\bgamma}T\}, \underline{\tilde{\blambda}})\notag\\
& = Pr(\underline{\tilde{\bgamma}}=\underline{\bgamma}T|\bW, \underline{\tilde{\bgamma}} \in \{\underline{\bgamma},\underline{\bgamma}T\}, \underline{\tilde{\blambda}})\label{eq:mimo1}\\
& = Pr(C=1|\bW, \underline{\bgamma}, \underline{\tilde{\blambda}}).\notag
\end{align}
Here \eqref{eq:mimo1} follows from Lemma~\ref{lem:perm_inv}. Thus we have
$I(C;\bW,\underline{\bgamma}, \underline{\tilde{\blambda}})=0$.
This proves that Alice does not gain any information about $C$ from
what she learns during the protocol.

We now discuss the optimal OT-pairing and the optimal power allocation.

\section{Optimization of the protocol}
\label{sec:optimize}
Let us first consider the simple setup where equal power is allocated in 
all pairs of subchannels, i.e.,
\begin{align*}
P_l = \frac{P}{2N} \phantom{xxx} \forall l.
\end{align*}
The capacity for this power allocation is
\begin{align*}
R = \sum_{l=1}^{N} \log \left(1+\frac{P\lambda_{\ell(l)}^{2}}{2N}\right)
       - \sum_{l=1}^{N}  \log \left(1+\frac{P\lambda_{k(l)}^{2}}{2N}\right)
\end{align*}
Clearly, this is maximized if $\lambda^2_{\ell(l)} > \lambda^2_{k(j)}$ for all $l,j$.
That is, provided the best half of the channels form the stronger
channels of the pairs, the achieved rate is independent of the
actual pairing. However, this is not true if we have the freedom to
pair the channels as well as to allocate variable power $P_l$ to different
pairs. In general, we would like to choose an {\em optimal pairing} $(\ell(l), k(l)); 
1\leq l \leq N$ and power allocation $P_l; 1\leq l \leq N$ so as
to maximize
\begin{equation}
R = \sum_{l=1}^{N}  \log \left(1+\frac{P_l\lambda_{\ell(l)}^{2}}{2N}\right)
    - \sum_{l=1}^{N}  \log \left(1+\frac{P_l\lambda_{k(l)}^{2}}{2N}\right).
\label{eq:opt1}
\end{equation}

The following theorem states that an optimal OT pairing couples the best
channel with the worst, and so on with the remaining channels.

\begin{theorem} \label{thm1}
An optimal pairing combines the best channel with the worst channel and continues
similarly with the remaining channels.  
That is, the pairing is given by $\ell(l)=\sigma(l)$ and $k(l)=\sigma(2N-l+1)$ 
for $l=1,\cdots,N$ for some permutation $\sigma$ which arranges
the gains in a non-increasing order.
\end{theorem}
The proof of the theorem is given in the appendix.
In the theorem, the permutation $\sigma$ is such that
$\lambda_{\sigma(l)}\geq \lambda_{\sigma(l+1)} \; \forall \; l<2N$.
This result reduces the problem of joint optimization of
\eqref{eq:opt1} for the best pairing and power allocation to
separate optimization of the pairing and the power allocation among the pairs of channels.
With high probability, all the gains $(\lambda_1,\cdots, \lambda_{2N})$ are distinct. 
Under this high probability event, Theorem~\ref{thm1} gives a unique optimal pairing. 
We now find the optimal power allocation.

\noindent
{\it Optimal Power Allocation:}
In light of Theorem \ref{thm1}, we assume that the channels are ordered such that 
\begin{equation*}
\lambda_{l} \geq \lambda_{l+1} \quad \mathrm{for} \; 1\leq l<2N
\end{equation*}
and the channel with gain $\lambda_l$ is paired with the channel with
gain $\lambda_l'$, where $\lambda_l' = \lambda_{2N-l+1}$.
Then for a given power allocation $P_l; 1\leq l\leq N$, the achieved rate is 
\begin{equation}
R(P_1,\cdots,P_N)=\sum_{l=1}^{N} \log (1+P_{l}\lambda_{l}^{2}) - \sum_{l=1}^{N} \log (1+P_{l}\lambda_{l}'^{2}). \label{eq:rate} \notag \\
\end{equation}
We need to maximize this with respect to the $P_l$s under the condition 
$$ \sum_{l=1}^{N}P_{l}\leq \frac{P}{2}. $$
Similar optimization was needed for power allocation over different fading states for block fading wiretap channel \cite{Gopala2008}.
This can be solved by defining the  Lagrangian objective function  
\begin{equation*}
 J=R(P_1,\cdots,P_N)-\eta\left( \sum_{l=1}^{N}P_{l} -\frac{P}{2} \right).
\end{equation*} 
%
The optimal power allocation is given by 
\begin{align*}
P_l & = \begin{cases}
         \left( \left( f(\lambda_{l},\lambda_{l}', \eta)\right) ^{1/2} -\frac{1}{2}\left( \frac{1}{\lambda_{l}^{2}} +\frac{1}{\lambda_{l}'^{2}} \right)  \right)^+ & \text{ if }\lambda_{l}' \neq 0\\
        \left(\frac{1}{\eta} - \frac{1}{\lambda_l^2}\right)^+ & \text{ if } \lambda_{l}' = 0
      \end{cases} 
\end{align*}
where 
$$f(\lambda_{l},\lambda_{l}', \eta) 
= \frac{1}{4}\left(\frac{1}{\lambda_{l}'^2} - \frac{1}{\lambda_{l}^2}\right)\left[
\left(\frac{1}{\lambda_{l}'^2} - \frac{1}{\lambda_{l}^2}\right) + \frac{4}{\eta}\right],$$\\
and $\eta$ is determined by the condition 
$$ \sum_{l=1}^{N}P_{l} = \frac{P}{2}. $$

%
%

\noindent
{\it Power allocation across coherence blocks:} If variable amount of average
power is allowed to be transmitted in different blocks under a long term
average power constraint, then potentially higher rates are achievable.
Let $(\lambda_1,\lambda_2,\cdots, \lambda_{2N})$ denote the random vector that represents
the ordered (non-increasing) channel vector in a block. The optimum 
pairing in each block 
is still as given by Theorem~\ref{thm1}. The optimal power allocation is
the maximizer of the expected rate
\begin{align*}
& R=E\left[\sum_{l=1}^{N}\left( \log (1+P_{l}(\underline{\blambda})\lambda_{l}^{2})
\right.\right. \\
& \hspace*{35mm}\left.\left. - \log (1+P_{l}(\underline{\blambda})\lambda_{2N-l+1}^{2})\right)\rule{0pt}{18pt}\right]
\end{align*}
under the average power constraint
\begin{equation*}
E\left[\sum_{l=1}^{N}P_{l}(\underline{\blambda})\right] \leq \frac{P}{2}.
\end{equation*}
By similar steps as before, the solution is given by
\begin{align*}
P_l(\underline{\blambda}) & = \begin{cases}
         \left( \left( f(\lambda_{l},\lambda_{l}', \eta)\right) ^{1/2} -\frac{1}{2}\left( \frac{1}{\lambda_{l}^{2}} +\frac{1}{\lambda_{l}'^{2}} \right)  \right)^+ & \text{ if }\lambda_{l}' \neq 0\\
        \left(\frac{1}{\eta} - \frac{1}{\lambda_l^2}\right)^+ & \text{ if } \lambda_{l}' = 0.
      \end{cases} 
\end{align*}
where $\eta$ is a global constant 
 determined by the condition
\begin{equation}
 E\left[\sum_{l=1}^{N}P_{l}(\underline{\blambda})\right] = \frac{P}{2}. \label{Pow_Con}
\end{equation}
Here $\eta$ depends only on the channel statistics and $P$.


\section{High SNR asymptotics}
\label{sec:high}
Let us consider a set of parallel channels. We want to study the
asymptotic expected rate. Let us consider a fixed ordered channel vector
$(\lambda_1, \lambda_2, \cdots, \lambda_{2N})$ to start with. Note that in the case
of a $(2N\times n_B)$ MIMO system with precoding, there are 
$2N$ channels. If $n_B\leq N$, then there are $n_B$ useful pairs of channels
with channel gains $(\lambda_1, \lambda'_{1}), (\lambda_2, \lambda'_{2}), \cdots, (\lambda_{n_B},\lambda'_{n_B})$,
where $\lambda'_{l}=\lambda_{2N-l+1} = 0$, for $l=1,2,\cdots, n_B$. If $N<n_B <2N$, then
there are $N$ pairs. $(2N-n_B)$ of them have the second channel gain
zero, more specifically, $\lambda'_{1}=\cdots = \lambda'_{(2N-n_B)} =0$.

Clearly, $\eta \rightarrow 0$ as $P\rightarrow \infty$. So, $P_l \rightarrow
\infty$ as $P\rightarrow \infty$. Now, for a pair of channels with
$\lambda'_{l} =0$, the rate contributed by the pair is\footnote{Here we mean $R_l - \log (P_l \lambda_l^2) \rightarrow 0$ as $P \rightarrow \infty$}
\begin{align}
R_l & = \log \left(1+P_l \lambda_l^2\right)\notag\\
& \rightarrow \log (P_l \lambda_l^2). \label{Eq:asy_power}
\end{align}

For such a channel pair, 
\begin{align}
P_l & = \frac{1}{\eta}\left(1 - \frac{\eta}{\lambda_l^2}\right)\notag\\
\Rightarrow \eta P_l & \rightarrow 1 \phantom{xxx} \text{as } 
\eta \rightarrow 0 \label{eq:asy1}
\end{align}
When $\lambda'_{l} \neq 0$ and $\lambda_l \neq \lambda'_{l}$, as  $\eta \rightarrow 0$,
\begin{align}
\sqrt{\eta} P_l & \rightarrow \left(\frac{1}{\lambda_{l}'^2} - \frac{1}{\lambda_{l}^2}\right)^{\frac{1}{2}}. \label{eq:asy3}
\end{align}
So, for such channel pairs,
\begin{align}
R_l & = \log \left(1+P_l \lambda_l^2\right) - \log \left(1+P_l \lambda_{l}'^2\right)\notag\\
& \rightarrow \log \left(\frac{\lambda_l^2}{\lambda_{l}'^2}\right) \phantom{xxx} \text{as }
P\rightarrow \infty .\label{Asym_Rate}
\end{align}
Now, using \eqref{eq:asy1} and \eqref{eq:asy3}, the power constraint gives
\begin{align}
&\eta P \rightarrow 2(2N-n_B) \label{eq:asy2}
\phantom{xx} \text{as } P\rightarrow \infty .
\end{align}
Inspired by similar concepts for communication over MIMO channels,
it is reasonable to define the {\it OT-multiplexing gain} as
\begin{align*}
\mu_{OT}=\lim_{P\rightarrow \infty} \frac{E\left[\sum_i R_i\right]}{\log P}.
\end{align*}

So, 
\begin{align}
\mu_{OT} & = \lim_{P\rightarrow \infty} \frac{E\left[\sum_{l:\lambda_{l}'=0} R_l\right]}{\log P} \quad \left(\text{using} \eqref{Asym_Rate}\right) \notag\\
& = \lim_{P\rightarrow \infty} \frac{E\left[\sum_{l:\lambda_{l}'=0} \log(P_l)\right]}{\log P} \quad \left(\text{using} \eqref{Eq:asy_power}\right) \notag\\
& = \lim_{P\rightarrow \infty} \frac{E\left[\sum_{l:\lambda_{l}'=0} (\log(P_l)-\log(\eta P_l))\right]}{\log P - E(\log (\eta P))}\label{eq:mux}\\
& = \lim_{P\rightarrow \infty} \frac{E\left[\sum_{l:\lambda_{l}'=0} (-\log(\eta))\right]}{- E(\log (\eta ))}\notag\\
& = E\left[|\{l:\lambda_{l}'=0\}|\right]\notag
\end{align}
Here \eqref{eq:mux} follows from \eqref{eq:asy1} and \eqref{eq:asy2}.
Thus our protocol achieves the OT-multiplexing gain of
\begin{align*}
\mu_{OT}& = \begin{cases}
n_B & \text{ if } n_B\leq N \\
2N-n_B & \text{ if } N<n_B \leq 2N\\
0 & \text{ if } n_B\geq 2N.
\end{cases}
\end{align*}

In contrast, for communication over a $2N\times n_B$ 
MIMO channel, the multiplexing gain is $\min\{n_B, 2N\}$.
For $n_B \geq 2N$, the average OT rate converges to a constant as $P \rightarrow \infty $.
This can be seen as a consequnce of the fact that
the secrecy capacity of the Gaussian wiretap channel goes to a constant as $P \rightarrow \infty $.

\section{Numerical results}
\label{sim:results}
In this section, we provide numerical results of our OT protocols for some simple MIMO and OFDM channels which include the examples discussed in 
Section~\ref{sec:scheme}. 

In Fig.~\ref{sim:MIMO}, we plot the OT rate of our protocol for $2\times 1$ and $2\times 2$ MIMO channels.
The average OT rate is numerically evaluated using Monte Carlo simulation methods for SNR varying from 0 dB to 50 dB.
The channel capacities for these channels with CSIT are also numerically evaluated and shown.
It can be seen that OT rate of $2 \times 1$ MIMO channel at SNR $P$ dB is approximately equal to the capacity 
of $2 \times 1$ MIMO channel with CSIT at 3 dB lower transmit power. 
This is due to the fact that in our OT protocol, half of the power is given to the null-space of ${\bf H}$
which is useless for communication. OT rate of
$2 \times 1$ MIMO channel increases at the rate of $1$ bit/3dB, as 
$\mu_{OT}=1$.

\begin{figure}[htbp]
\centering
\includegraphics[height=50mm,width=\columnwidth]{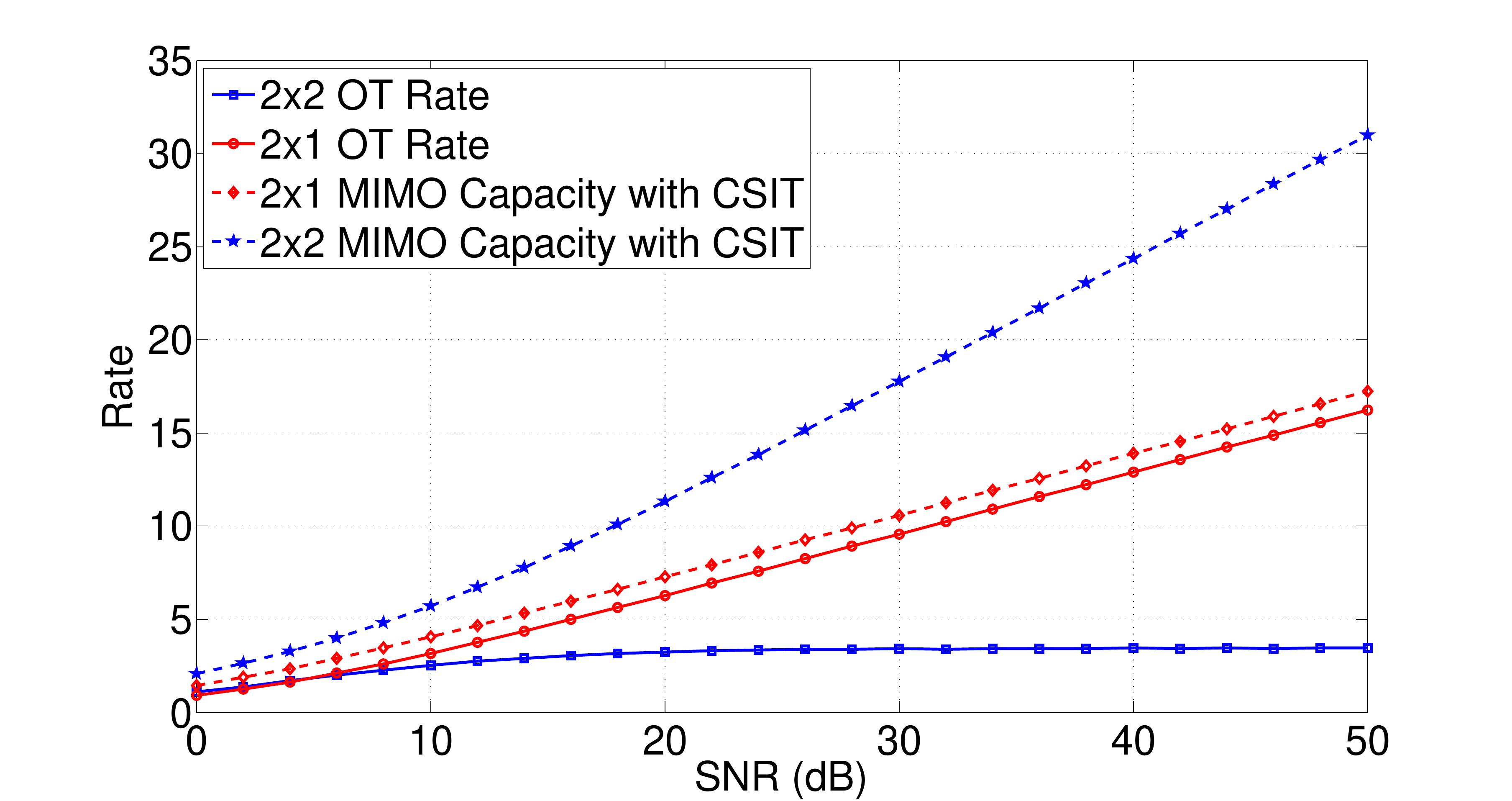}
\caption{OT Rate and MIMO capacity versus SNR for $2\times 1 ,2\times 2$ MIMO}
\label{sim:MIMO}
\end{figure}

Using \eqref{Asym_Rate} we see that at very high SNR, the OT rate
for $2\times 2$ MIMO system is given by
$R \approx E\left[\log \left(\frac{\lambda_0^2}{\lambda_1^2}\right)\right] $. 
Recall that $\lambda_0^2, \lambda_1^2$ are the eigenvalues of the Wishart matrix ${\bf HH}^{\dagger}$. 
The joint p.d.f. of the ordered eigenvalues, $\gamma_0=\lambda_0^2, \gamma_1
=\lambda_1^2$, is given by 
$e^{-(\gamma_0+\gamma_1)}(\gamma_0-\gamma_1)^2$ \cite[Theorem~2.17]{Tulino2004}.
The asymptotic value of the OT rate is thus
\begin{align*}
E\left[\log\left(\frac{\gamma_0}{\gamma_1}\right)\right] 
 & =  \int\limits_0^{\infty} \int\limits_0^{\gamma_0} \log\left(\frac{\gamma_0}{\gamma_1}\right) e^{-(\gamma_0+\gamma_1)}(\gamma_0-\gamma_1)^2 d\gamma_1 d\gamma_0 \\
 & =  1+ 2 \ln (2) \; \textrm{nats} \approx 3.45 \;\textrm{bits}.
\end{align*}

\begin{figure}[htbp]
\centering
\includegraphics[height=50mm,width=1\columnwidth]{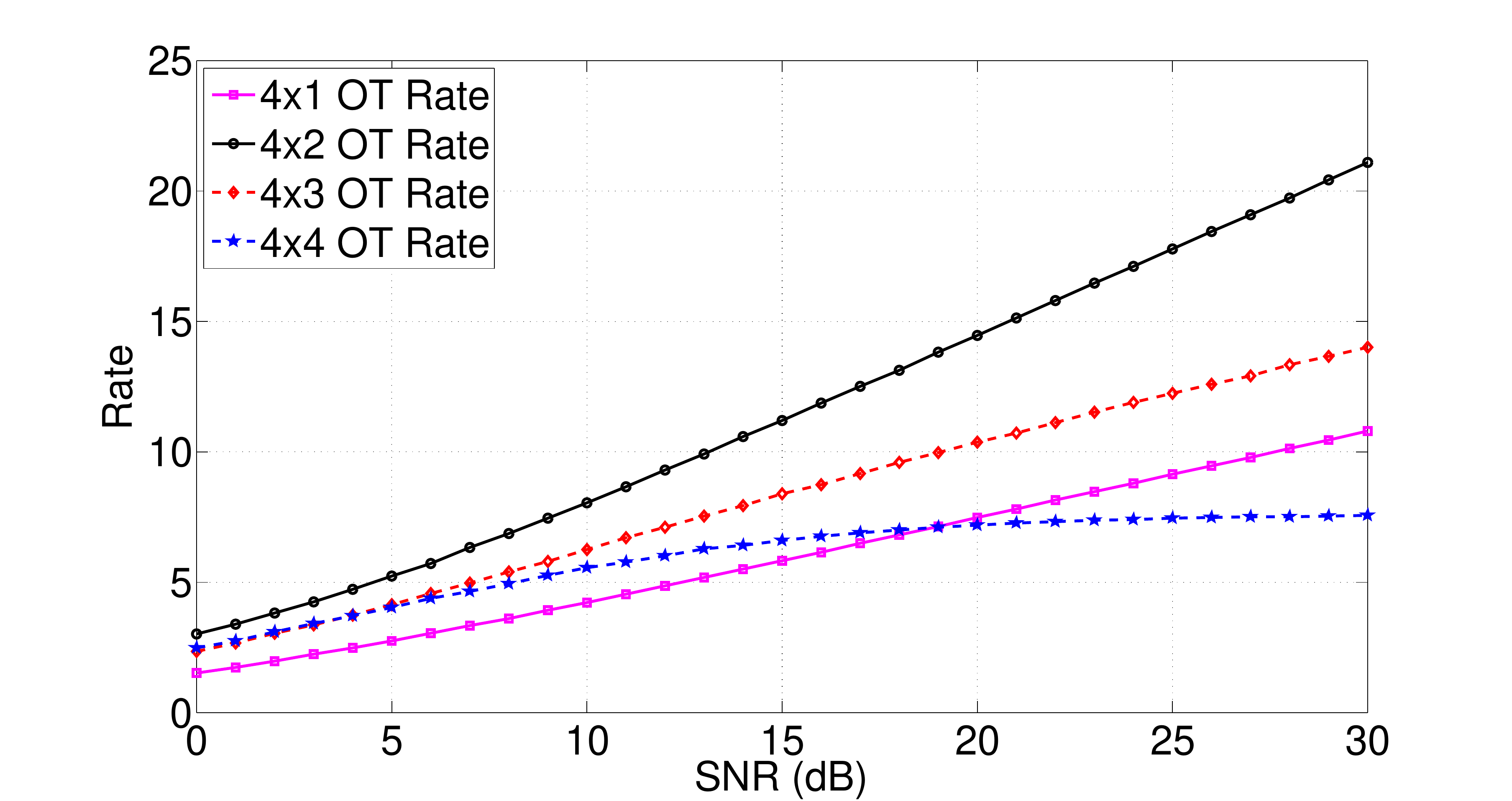}
\caption{OT Rates for MIMO with $n_A=4$ transmit antennas, and $n_B=1,2,3,4$ receive antennas}
\label{sim:MIMO4by4}
\end{figure}

In Fig.~\ref{sim:MIMO4by4}, OT rates for MIMO with $n_A = 4$ and $1 \leq n_B \leq 4$ are shown as a function of SNR. As expected from
 Section~\ref{sec:high}, the best OT rate is achieved 
when $n_B=n_A/2=2$, with asymptotic slope of 2 bits/3dB ($\mu_{OT} = 2$). 
The asymptotic slope for $n_B=1$ and $n_B=3$ is 1 bit/3dB ($\mu_{OT}=1$).
For $n_B\geq 4$, $\mu_{OT}=0$, and the rate  is bounded.

\begin{figure}[htbp]
\centering
\includegraphics[height = 50mm,width = \columnwidth]{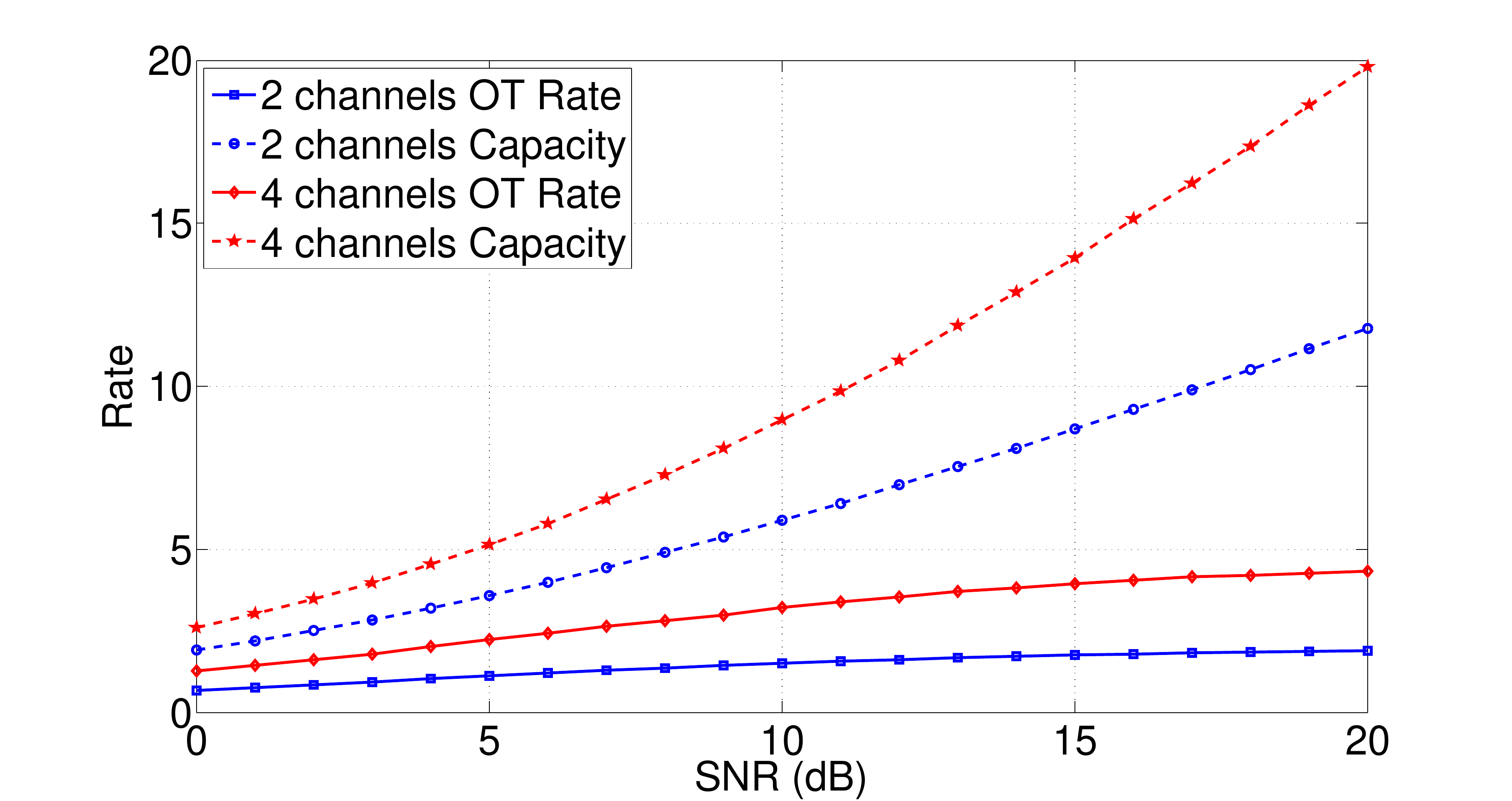}
\caption{OT Rate and OFDM capacity versus SNR for 2, 4 Channels OFDM}
\label{sim:OFDM} 
\end{figure}

In Fig.~\ref{sim:OFDM}, we show the OT rate for 2-channels OFDM and 4-channels OFDM, along
with the capacities of the corresponding channels. 
The OT rate of 2-channel OFDM converges to a constant as SNR 
increases, since $\mu_{OT}=0$. To find this constant, we note that $|H_0|$ 
and $|H_1|$ are i.i.d. with Rayleigh distribution. So
$|H_0|^2$ and $|H_1|^2$ have exponential distribution. 
Let $S =  \max(|H_0|^2, |H_1|^2)$ and 
$T = \min (|H_0|^2, |H_1|^2)$. Then the probability density functions of $S$ and $T$
are $2(1-e^{-s})e^{-s}$ and $2e^{-2t}$ respectively.
As SNR increases, the OT rate for our protocol converges to 
\begin{align*}
E[\log(S/T)] & = \int\limits_0^{\infty} \int\limits_0^{\infty} \log(s/t) 2(1-e^{-s})e^{-s} 2e^{-2t} ds dt \\
& = 2 \ln (2)\;\text{nats} = 2 \; \text{bits}.
\end{align*}
The OT rate of 4-channels OFDM also converges to a constant and $\mu_{OT}=0$.

\section{Conclusion}
\label{sec:conc}
We presented a technique for OT over parallel fading AWGN channels
with receiver CSI with application to OFDM and MIMO. 
For privacy of Bob against Alice, our techniques use primarily Bob's 
exclusive knowledge of the fading states, whereas the additive noise
is utilized for privacy of Alice against Bob. 

In AWGN channels,
the noise realization is used to perform OT in~\cite{Isaka, IsakaJournal}.
Following similar principle, the noise realization
can potentially be further utilized in our setup to achieve better rate.
In particular, for a single point-to-point fading channel or for parallel
fading channels with the same fading coefficient,
an obvious scheme is for Bob to first reveal the 
channel state to Alice over the noise-free channel. Then they can follow a
protocol suitable for the resulting AWGN channel. 
However, as pointed out in~\cite{IsakaJournal}, the OT rate saturates
to a constant as $P\rightarrow \infty$ in AWGN channels. Thus further
utilization of the noise realization in our protocol will not only result
in a much more complex protocol, but it will also not provide any additional asymptotic
OT-multiplexing gain. 

With an odd number of OFDM channels, or an odd number of transmit
antennas in a MIMO system, we have an odd number of parallel
channels. In such a case, our protocol will leave one channel of
middle rank in strength unused. That channel-state can be revealed to Alice
by Bob, and the OT protocol of \cite{IsakaJournal} can be used in the
resulting AWGN channel. This also does not give any asymptotic ($P\rightarrow
\infty$) improvement in terms of multiplexing gain.

Altogether, the technique
proposed in this paper can be an important tool for performing OT efficiently
over wireless channels.

\appendices
\section{Proof of Theorem~\ref{thm1}}
\label{sec:app1}
\begin{lemma}   \label{lemma1}
If $P_1>P_2$, $\alpha>\beta$, then $(1+P_1\alpha)(1+P_2\beta) > (1+P_1\beta)(1+P_2\alpha)$. 
\end{lemma}

\begin{proof}
We first note the following basic fact.

{\it Claim:} If $x,y>0$ , $\frac{x}{y}>1$, then $f(\alpha) = \frac{x+\alpha}{y+\alpha}$ is a monotonically
decreasing function of $\alpha$.

{\it Proof of the claim:}
 It can be easily checked that $\frac{df}{d\alpha} = \frac{y-x}{(y+\alpha)^2} < 0 \; \; \forall \alpha$.
 Thus the claim follows.

Now by the hypothesis of the lemma, $\alpha>\beta$ and $\frac{1}{P_1}< \frac{1}{P_2}$. Thus by the above claim,
 
 \begin{align*}
  \frac{\alpha+\frac{1}{P_2}}{\beta+\frac{1}{P_2}} &< \frac{\alpha+\frac{1}{P_1}}{\beta+\frac{1}{P_1}}\\
  \implies \frac{\alpha P_2+1}{\beta P_2+1} &< \frac{\alpha P_1+1}{\beta P_1+1}\\
  \implies (1+\alpha P_1 )(1+\beta P_2) &> (1+\alpha P_2)(1+\beta P_1)
 \end{align*} 
\end{proof}

\begin{lemma}  \label{lemma2}
 For any $l,j \in \{1,2,\dots, N\}$, an optimal protocol can not have $\lambda_{\ell(l)} > \lambda_{k(l)}> \lambda_{\ell(j)}> \lambda_{k(j)}$.
\end{lemma}
\begin{proof}
We will show that under the above condition,  the pairing can be improved strictly with the same power allocation.
Let us consider another pairing defined by $\ell'(\cdot),k'(\cdot)$ such that

\begin{eqnarray*} 
 \ell'(t) &=& \ell(t) \; \; \forall t\neq j\\
 \ell'(j) &=& k(l)\\ 
 k'(t) &=&k(t) \: \: \forall t\neq l\\
 k'(l) &=& \ell(j) 
\end{eqnarray*}

That is, $k(l)$ and $\ell(j)$ are interchanged. Clearly $\ell',k'$ define a valid pairing. Consider the same power allocation. 
Only the rates $R_l,R_j$ will change to $R_l',R_j'$ (say).

\begin{eqnarray*} 
 R_l+R_j &=& \log \frac{(1+P_l\lambda_{\ell(l)}^2)(1+P_j\lambda_{\ell(j)}^2)} {(1+P_l\lambda_{k(l)}^2)(1+P_j \lambda_{k(j)}^2)}\\
 R_l'+R_j' &=&  \log \frac{(1+P_l\lambda_{\ell(l)}^2)(1+P_j\lambda_{k(l)}^2)} {(1+P_l\lambda_{\ell(j)}^2)(1+P_j \lambda_{k(j)}^2)}\\
 (R_l+R_j) - (R_l'+R_j') &=& \log \frac{(1+P_j\lambda_{\ell(j)}^2)(1+P_l\lambda_{\ell(j)}^2)}{(1+P_l\lambda_{k(l)}^2)(1+P_j\lambda_{k(l)}^2)} \\
& < & 0 \; (\mbox{since} \: \lambda_{\ell(j)}<\lambda_{k(l)} )
\end{eqnarray*}

Thus $R_l'+R_j'> R_l+R_j$ . Since $R_t= R_t' \; \forall t \neq l,j$, the new pairing gives more rate with the same
power allocation.
\end{proof}
\begin{lemma} \label{lemma3} 
For an optimal protocol 
\begin{equation*}
\lambda_{\ell(l)}\geq \lambda_{k(j)} \quad \forall \; l,\;  j.
\end{equation*}
\end{lemma}
\begin{proof} If this is not true, then suppose
\begin{equation*}
\lambda_{\ell(l)} < \lambda_{k(j)} \quad \mathrm{for \; some}\quad l, j
\end{equation*}
Then 
\begin{equation*}
\lambda_{\ell(j)} > \lambda_{k(j)}> \lambda_{\ell(l)} > \lambda_{k(l)}
\end{equation*}
which can not be true by Lemma~\ref{lemma2}.
\end{proof}

\begin{lemma} \label{lemma4} 
For an optimal protocol 
\begin{equation*}
\lambda_{\ell(l)}> \lambda_{\ell(j)} \Rightarrow  \lambda_{k(l)} \leq  \lambda_{k(j)}.
\end{equation*}
\end{lemma}
\begin{proof} By contradiction, suppose $l, \; j$ are such that
\begin{equation*}
\lambda_{\ell(l)} > \lambda_{\ell(j)} \quad \mathrm{and}\quad \lambda_{k(l)} > \lambda_{k(j)}
\end{equation*}
\begin{equation*}
\Rightarrow \quad \lambda_{\ell(l)} > \lambda_{\ell(j)}> \lambda_{k(l)} > \lambda_{k(j)}
\end{equation*}
as $\quad \lambda_{\ell(l)} > \lambda_{k(l)}> \lambda_{\ell(j)} > \lambda_{k(j)} $
can not be true by Lemma~\ref{lemma2}.

\noindent \textbf{Case 1: } $P_{l} > P_{j}$ \\
By Lemma \ref{lemma1}, 
\begin{align}
& \log (1+P_{l}\lambda_{k(l)}^{2}) + \log (1+P_{j}\lambda_{k(j)}^{2}) \notag \\
& \phantom{xxxx} >  \log (1+P_{j}\lambda_{k(l)}^{2}) + \log (1+P_{l}\lambda_{k(j)}^{2}) \label{eq:1} 
\end{align}
Consider a different pairing $l, \;k'$ such that
\begin{equation*}
  k'(t)=\left\{
\begin{array}{cl}
 k(t) &  ; \; t \neq l, j \\
              k(l) & ; \;  t=j \\
              k(j) & ;\; t=l 
\end{array} \right. 
\end{equation*}
i.e. $k(l),\; k(j)$ are interchanged. 
Then the new rate $R'$ is such that
\begin{eqnarray*}
 R'-R & = & \sum_{t=1}^{N} (R'_{t}-R_{t}) \\
 &= & (R'_{l}-R_{l}) + (R'_{j}-R_{j}) \\
 & = & (R'_{l}+R'_{j}) - (R_{l}+R_{j}) \\
 & = & \left[  \log (1+P_{l}\lambda_{\ell(l)}^{2}) - \log (1+P_{l}\lambda_{k(j)}^{2}) \right. \\
 &  &  + \left. \log (1+P_{j}\lambda_{\ell(j)}^{2}) - \log (1+P_{j}\lambda_{k(l)}^{2}) \right]  \\
 & - & \left[  \log (1+P_{l}\lambda_{\ell(l)}^{2}) -\log (1+P_{l}\lambda_{k(l)}^{2}) \right. \\
 &  & +\left. \log (1+P_{j}\lambda_{\ell(j)}^{2}) - \log (1+P_{j}\lambda_{k(j)}^{2}) \right] \\
 & > & 0 \quad \mathrm{by}  \; \eqref{eq:1}.
\end{eqnarray*} 
Thus the new pairing strictly improves the rate. \\
\noindent \textbf{Case 2: } $P_{l} < P_{j}$ \\
By Lemma \ref{lemma1}, 
\begin{align}
& \log (1+P_{j}\lambda_{\ell(l)}^{2}) + \log (1+P_{l}\lambda_{\ell(j)}^{2}) \notag \\
& \phantom{xxxxxx} >  \log (1+P_{l}\lambda_{\ell(l)}^{2}) + \log (1+P_{j}\lambda_{\ell(j)}^{2}) \label{eq:2}
\end{align}
Consider a different pairing $\ell', \;k$ such that
\begin{equation*}
  \ell'(t)=\left\{
\begin{array}{cl}
 \ell(t) &  ; \; t \neq l, j \\
              \ell(l) & ; \;  t=j \\
              \ell(j) & ;\; t=l 
\end{array} \right. 
\end{equation*}
i.e. $\ell(l),\; \ell(j)$ are interchanged. 
Then the new rate $R'$ is such that
\begin{eqnarray*}
 R'-R & = &  (R'_{l}+R'_{j}) - (R_{l}+R_{j}) \\
 & = & \left[  \log (1+P_{j}\lambda_{\ell(l)}^{2}) + \log (1+P_{l}\lambda_{\ell(j)}^{2}) \right. \\
 &  & -  \left. \log (1+P_{l}\lambda_{\ell(l)}^{2}) - \log (1+P_{j}\lambda_{\ell(l)}^{2}) \right]  \\
 & > & 0 \quad \mathrm{by}  \; \eqref{eq:2}.
\end{eqnarray*} 
So the new pairing strictly improves the rate. This completes the proof of the lemma. 
\end{proof}

Now let us assume, without loss of generality, that the pairs are indexed such
that
\begin{equation}
\lambda_{\ell(l)}\geq \lambda_{\ell(l+1)} \quad \forall \; l=1,2,\cdots ,N \label{eq:3}
\end{equation}
and 
\begin{align}
& \lambda_{k(l)}\leq \lambda_{k(l+1)} \quad \mathrm{whenever} \quad \lambda_{\ell(l)}=\lambda_{\ell(l+1)}\label{eq:4} 
\end{align}
for  $l=1,2,\cdots ,N.$ 

\noindent
{\it Proof of Theorem~\ref{thm1}:} Let us define
\begin{equation}
 \sigma(l)=\ell(l) \quad \mathrm{for} \;  l=1,\cdots,N. \label{eq:5} \notag
 \end{equation} 
and
 \begin{equation}
 \sigma(l)=k(2N-l+1) \quad \mathrm{for} \;  l=N+1,\cdots,2N. \label{eq:6} \notag
 \end{equation} 
 We now need to prove that $\lambda_{\sigma(l)}\geq \lambda_{\sigma(l+1)} \quad \forall \; l$.\\
 For $l=1,2,\cdots,N-1$, this follows from \eqref{eq:3}. 
 For $l=N$, this follows from Lemma \ref{lemma3}. 
 For $N< l< 2N$, if $\lambda_{\sigma(l)}< \lambda_{\sigma(l+1)}$, then 
 \begin{equation*}
 \lambda_{k(j-1)} > \lambda_{k(j)} \quad \mathrm{where} \; j=2N-l+1 >1
 \end{equation*}
 But then 
 \begin{equation*}
\lambda_{\ell(j-1)} \geq \lambda_{\ell(j)}\geq \lambda_{k(j-1)} > \lambda_{k(j)} 
\end{equation*}
 This contradicts either \eqref{eq:4} or Lemma \ref{lemma4}. 
 Thus it must be true for $N< l< 2N$ that 
 \[ \lambda_{\sigma(l)} \geq \lambda_{\sigma(l+1)} \]
This completes the proof of the Theorem.

\end{document}